%% file: main.tex
\documentclass[final]{jfp-epi}
\jfpVolume{36}
\jfpArticle{9}
\jfpDOI{10.46298/jfp.17808}
\jfpYear{2026}
\received[Submitted]{July 2025}
\received[accepted]{June 2026}

\usepackage[utf8]{inputenc}
\usepackage[T1]{fontenc}
\usepackage{graphicx}
\usepackage{listings, textcomp, xcolor, soul}
\definecolor{HlGreen}{rgb}{0.988, 0.906, 0.275}
\sethlcolor{HlGreen}

\usepackage{xspace}

\usepackage{bbold, bbding}
\newcommand{\mathbbm}[1]{\mathbb{#1}}
\usepackage{multicol, tikz}
\usetikzlibrary{arrows,automata,positioning}

\newtheorem{remark}{Remark}

\input{macros.tex}
\newtheorem{theorem}{Theorem}[section]
\newtheorem{lemma}[theorem]{Lemma}
\newtheorem{example}[theorem]{Example}
\newtheorem{corollary}[theorem]{Corollary}

\begin{document}

\title{A simple and efficient implementation of strong~call~by~need by an abstract machine}

\author{Ma{\l}gorzata Biernacka}
 \orcid{0000-0001-8094-0980}
 \affiliation{%
   \institution{University of Wrocław, Faculty of Mathematics
     and Computer Science}
   \city{Wrocław}
   \country{Poland}
   \authoremail{mabi@cs.uni.wroc.pl}
}

 \author{Witold Charatonik}
 \orcid{0000-0001-7062-0385}
 \affiliation{%
   \institution{University of Wrocław, Faculty of Mathematics and Computer Science}
   \city{Wrocław,}
   \country{Poland}
  \authoremail{wch@cs.uni.wroc.pl}
}

 \author{Tomasz Drab}
 \orcid{0000-0002-6629-5839}
 \affiliation{%
   \institution{University of Wrocław, Faculty of Mathematics and Computer Science \& Google}
   \city{Wrocław}
   \country{Poland}
  \authoremail{tdr@cs.uni.wroc.pl \& tomaszdrab@google.com}
 }

\begin{abstract} 
  Strong call-by-need combines full normalization with the sharing
  discipline of lazy evaluation, yet no prior implementation achieved
  both simplicity and efficiency.  We introduce RKNL, an abstract
  machine that realizes strong call-by-need with bilinear
  overhead. The machine has been derived automatically from a
  higher-order evaluator that uses the technique of memothunks to
  implement laziness. By employing an off-the-shelf transformation
  tool implementing the ``functional correspondence'' between
  higher-order interpreters and abstract machines, we obtained a
  simple and concise description of the machine.  We prove that the
  resulting machine conservatively extends the lazy version of the Krivine
  machine for the weak call-by-need strategy, and that it simulates
  the normal-order strategy in a bilinear number of steps, i.e.,
  linear in both the number of $\beta$-reductions and the size of the
  input term.
\end{abstract}






\maketitle

\section{Introduction}
Lambda calculus is a well-established model for functional programming
languages. Practical implementations rely on a deterministic
reduction strategy (also known as evaluation strategy), which restricts
the general $\beta$-reduction rule to specific positions 
within a term. Popular incarnations of the lambda calculus include the
full-fledged programming languages OCaml and Haskell, each of which
adopts a different evaluation strategy: weak call-by-value (CbV) and
weak call-by-need (CbNd), respectively.

A strategy is called \emph{weak} if it does not descend into and
reduce under lambda abstractions; hence, any lambda abstraction is
treated as a value. The most widely used strategy in practical
programming is weak CbV, where each argument is evaluated exactly
once, before being passed to the function, even if its result is not
used in the function body. In contrast, weak call-by-name (CbN)
performs $\beta$-reduction without evaluating the argument
beforehand. Consequently, an argument may be evaluated multiple times
if it is \emph{needed} more than once, i.e., if the variable it
replaces occurs more than once in the function body in positions not
eliminated by $\beta$-reduction. To mitigate this inefficiency,
practical implementations (such as Haskell) employ \emph{call-by-need}
evaluation, which lazily evaluates arguments and memoizes their
results for reuse.  In this strategy each argument is evaluated at
most once, only if it is needed.
  

In some applications, weak strategies are insufficient. The
development of proof assistants based on dependent type theory (such
as Rocq or Agda) has prompted deeper investigation into \emph{strong}
reduction in the lambda calculus and its implementations. Strong
reduction computes full normal forms of terms (required for type
checking in proof assistants), necessitating reductions inside lambda
abstractions and handling of open terms. A natural way to obtain a
strong strategy from a weak one is to first weakly reduce the term to
a value, then recursively reduce within function bodies and argument
positions of \emph{rigid terms}—terms of the form $x\ t_1 \ldots t_n$,
where the variable $x$ will never be substituted and thus cannot
create new $\beta$-redexes. Strong strategies defined this way are
conservative extensions of their weak counterparts.

In this work, we show the first efficient abstract machine for Strong
CbNd: the RKNL machine. It is derived in a mechanical way using
functional programming techniques, and we give its formal account. Our
goal is twofold: (a) to provide a simple and concise definition of an
abstract machine for Strong CbNd, and (b) to facilitate reasoning
about the key properties of the machine and its underlying strategy,
including its correctness and complexity. To achieve these goals, we
adopt the derivational approach of~\citet{Biernacka-al:PPDP21},
previously applied to Strong CbV. We begin with a standard
higher-order normalization-by-evaluation function known to correspond
to normal-order reduction~\cite{BiernackaBCD20}. This normalizer is
then optimized by the standard technique of memoization to avoid
recomputation of intermediate values. Finally, we apply the
\emph{functional correspondence} transformation, as implemented in
Racket by~\citet{BuszkaB21:LOPSTR21}, to systematically derive a
concise abstract machine. We prove that the resulting RKNL machine
correctly simulates normal-order reduction, and as a measure of its
efficiency we give a simple proof that it is {\em reasonable} in the
sense of Accattoli et al.
\cite{Accattoli-DalLago:LMCS16,DBLP:conf/lics/AccattoliCC21,DBLP:journals/scp/AccattoliG19}. Our
development illustrates a principled and versatile approach to
deriving abstract machines from well-understood high-level
semantics. To reason about efficiency, we apply a variant of amortized
cost analysis using a potential function~\cite{Okasaki:99}.

The main contributions of this work are as follows:
\begin{enumerate}
\item A novel abstract machine for Strong CbNd, with a simple and clear design (only 11
  transition rules), derived automatically with a
  semantic transformer of \citet{BuszkaB21:LOPSTR21} from an NbE
  normalizer.
\item A proof that the machine is reasonable for time, i.e., that the number of steps performed by the machine is bounded by a
  polynomial---in this case bilinear---function in the number of $\beta$-steps
  and in the size of the initial term.
\item A proof that the machine conservatively extends a weak
  call-by-need machine.
\item An improvement in simulation overhead for Strong CbNd over the existing machine for this strategy~\cite{BiernackaC19},
  from exponential (unreasonable) to polynomial (reasonable). These machines are currently the only artefacts formalizing the strong call-by-need strategy, and amenable to complexity analysis.
\item An improvement of the overall simulation overhead of normal
  order by Useful MAM~\cite{Accattoli16} from
  quadratic~\cite{DBLP:conf/lics/AccattoliCC21} to quasibilinear in
  the RAM model.
\end{enumerate}

\paragraph*{Motivation.}
The study of call-by-need calculi is generally motivated by the quest
for efficient implementations of $\beta$-reduction in the
$\lambda$-calculus. It is worth noting that our contribution here is
primarily \emph{theoretical}. We start from a well-understood,
high-level specification (an NbE normalizer for normal order) and show
how to derive—almost mechanically—a compact abstract machine whose
correctness and complexity can be reasoned about with pen-and-paper
proofs, or it can be implemented or formalized with reasonable
effort. Furthermore, this approach can provide the basis to analyse
other strategies along similar lines, allowing us to derive semantic
artefacts that are correct by construction. The emphasis is on
principles: a transparent derivation pipeline, a small rule set, and a
clean complexity analysis; we prove the bilinear bound for machine
steps, which guarantees that the normal form of a lambda term can be
computed in time subquadratic in the number of normal-order reduction
steps. A unit of time in our analysis is also abstract---it is a
machine step, rather than, say, a millisecond as in realistic
benchmarking. In this sense, our work sits closer to the semantic and
proof-oriented lines of
\citet{Balabonski-al:ICFP17,Balabonski-al:FSCD21,BiernackaC19,Accattoli16,DBLP:conf/lics/AccattoliCC21}
than to large, hand-engineered evaluators. Furthermore, because a
compact, analysable CbNd machine yields a practical \emph{cost model}
for the $\lambda$-calculus, our work links naturally to the broader
agenda of Implicit Computational Complexity and to recent efforts that
treat the $\lambda$-calculus as a manageable computational model for
complexity-theoretic formalizations~\cite{GaeherKunze:2021:Cook-levin,Forster-al:ITP2021}.

\paragraph*{Related work.}
Given their practical motivation, strong strategies have been the
target of efficient implementation efforts in the form of abstract
machines. The first such machines were proposed by Crégut: the KN
machine—a strongly reducing extension of the Krivine machine that
performs normal-order reduction—and its lazy variant
KNL~\cite{Cregut:HOSC07}, which introduces sharing of head normal
forms but does not achieve full normalization, and thus is not strong
in our sense. The KN machine was later deconstructed into a
normalization-by-evaluation (NbE) function for call-by-name
semantics~\cite{Munk:MS,BiernackaBCD20}. More recently, Accattoli et
al. have explored strong strategies within the Linear Substitution
Calculus, developing abstract machines with complexity
bounds~\cite{Accattoli-Barras:PPDP17,DBLP:journals/scp/AccattoliG19,Accattoli-Coen:LICS15}. Parallel
to this, a derivational approach initiated by Danvy et
al.~\cite{Ager-al:PPDP03,Ager-al:IPL04,Pirog-Biernacki:Haskell10} has
been extended to strong strategies, enabling systematic connections
between high-level semantics and abstract machines
\cite{BChZ-fscd17,BiernackaC19,BiernackaBCD20}.

The canonical implementation of weak call-by-need is the STG machine
underlying Haskell~\cite{PeytonJones:JFP92}, which uses a store to
memoize computed values. This machine has been systematically derived
from a natural semantics for the STG
language~\cite{Pirog-Biernacki:Haskell10}, aligning with the
store-based formal semantics of lazy languages pioneered by
Launchbury~\cite{Launchbury:POPL93} and further refined by
Sestoft~\cite{Sestoft:JFP97}. An alternative, purely syntactic
approach without an explicit store was proposed
by~\citet{Ariola-al:POPL95}. In this tradition lie recent efforts to
formalize the Strong CbNd
strategy~\cite{Balabonski-al:ICFP17,Balabonski-al:FSCD21,Barenbaum-al:PPDP18,BiernackaC19}. In
particular, \citet{Balabonski-al:ICFP17} proposed a reduction semantics
for Strong CbNd using a calculus with explicit substitutions. More
recently, they refined this into a strategy that yields shorter
normalizing sequences than previously known
calculi~\cite{Balabonski-al:FSCD21}. This semantics was
operationalized into a refocusable reduction semantics, from which an
abstract machine was derived by~\citet{BiernackaC19}. While this
machine precisely implements the corresponding strategy, it operates
over a language with explicit substitutions or let-constructors for
sharing and is not directly optimized for efficient execution—it
contains 24 transition rules and lacks constant-time
transitions. Meanwhile, the Rocq proof assistant includes a lazy and
efficient machine for strong reduction, though it is a large, complex
artefact whose strategy remains to be formally characterized.

As a benchmark for efficiency, Accattoli’s Useful
MAM~\cite{Accattoli16} is both simple (10 transition rules) and
\emph{reasonable}: it simulates normal-order (strong call-by-name)
reduction with a number of steps quadratic in the number of
$\beta$-steps and linear in the input size. A recent improvement is
the SCAM machine~\cite{DBLP:conf/lics/AccattoliCC21}, which simulates
strong call-by-value with a bilinear bound and comprises only 9 rules,
although it requires precompilation. However, no similarly reasonable
abstract machine has been proposed for strong call-by-need. The
difficulty lies not only in the strength of the strategy but also in
the complex sharing behaviour required by call-by-need
evaluation~\cite{DBLP:conf/csl/AccattoliL22}.

To further position this work, let us consider the Higher-order
Virtual Machine (HVM)~\cite{hvm}.  A visible difference is that HVM's
computation model is the Interaction Nets \cite{Lafont97}, while our
machine works directly on the $\lambda$-calculus syntax. Furthermore,
HVM aims at performance measured with benchmarks and compared to the
Glasgow Haskell Compiler (GHC).  Our work does include empirical
measurements, but only as supporting evidence for a principled
complexity story. In short, HVM is an engineering artefact optimized
for speed, while RKNL is a minimal core with asymptotic guarantees.

A separate line of (theoretical) research investigates \emph{full
laziness}, or \emph{fully lazy
sharing}~\cite{Balabonski:POPL12,DBLP:conf/fscd/AccattoliMPC25}.
This approach aims for even
more sharing than call-by-need, but at the cost of significantly more
complex calculi, and has so far treated \emph{weak} evaluation
only. In fully lazy sharing, values are further split into a part that
has to be duplicated, and a part that can be shared. By contrast, we
study \emph{strong} evaluation and relate our results to normal order
in the pure $\lambda$-calculus. Extending our approach to encompass
fully lazy sharing is an interesting direction for the next step in
our research program.

The comparison of RKNL with existing machines and calculi for Strong
CbNd is subtle, and could be further investigated
(cf. Figure 13 in \cite{Danvy-Zerny:PPDP13}). In summary, we can state
the following main differences with respect to RKNL: 
\begin{itemize}
\item Crégut's KN from
\cite{Cregut:HOSC07} performs normal-order reduction, but is not
call-by-need and suffers from the exponential overhead
\item Crégut's KNL from \cite{Cregut:HOSC07} is strong and call-by-need, but performs
only head reduction instead of full reduction.
\item
\citet{Balabonski-al:FSCD21} gives a calculus and a strategy that
performs optimization similar to KNL, but is not an abstract machine.
\item
the machine from \cite{BiernackaC19} implements
earlier calculus from \cite{Balabonski-al:ICFP17},
and also has exponential overhead.
\item the machine underlying the Rocq prover is a complex artefact handling full Calculus of Inductive Constructions; its foundational properties have been studied in a simpler setting by \citet{Accattoli-Barras:PPDP17}.
\end{itemize}

\paragraph*{Outline.}
In Section~\ref{sec:prelim}, we recall notation used in the rest of
the paper to reason about reduction strategies. In
Section~\ref{sec:derivation}, we describe how our abstract machine is
derived from a~higher-order NbE normalizer. In
Section~\ref{sec:machine}, we present the machine and discuss its
transitions. We also show an elaborate example of its execution as
well as examples of empirical execution lengths for several term
families. In Section~\ref{sec:proofs}, we discuss correctness of the
derived machine. We prove its soundness and completeness, we show that
it simulates normal order in bilinear number of steps and that it
conservatively extends a~weak CbNd machine. In
Section~\ref{sec:conclusion}, we conclude.

\section{Preliminaries}
\label{sec:prelim}
Terms $t$ in the lambda calculus are defined
with the following grammar:
\begin{alignat*}{3}
t &::= x \alt \tapp{t_1}{t_2} \alt \tlam x t
\end{alignat*}
where $x$ ranges over an infinite, countable set of identifiers. As
usual, we assume that application is left-associative and has higher
precedence than abstraction, so that $\tlam x {\tapp{\tapp{x}{y}}{z}}$
is $\tlam x {(\tapp{(\tapp{x}{y})}{z})}$. We define free and bound
variables in a term in the usual way.  We say that two terms $t_1$ and
$t_2$ are $\alpha$-equivalent, and we then write $t_1 =_{\alpha}t_2$,
if $t_1$ and $t_2$ are equal up to renaming of bound variables.

Contexts can be seen as terms with exactly one ``hole''
(denoted $\hole$) which can occur in any position within a term:
\begin{alignat*}{3}
C &::= \hole \alt \tapp C t \alt \tapp t C \alt \tlam x C
\end{alignat*}
Given a context $C$ and a term $t$ to {\em plug} in its hole, we can
reconstruct the intended term (denoted $\plug C t$) by defining the plugging
function as follows:
\begin{center}
$\plug {\hole} s = s$ \hfill
$\plug {(\tapp C t)} s = \tapp {\plug C s} t$ \hfill
$\plug {(\tapp t C)} s = \tapp t {\plug C s}$ \hfill
$\plug {(\tlam x C)} s = \tlam x {\plug C s}$
\end{center}

\subsection{Reduction semantics in the lambda calculus}
\label{sec:reduction}
Computation in the lambda calculus consists in performing $\beta$-contraction
$$\tapp {(\tlam x {t_1})} {t_2} \; \rightharpoonup_\beta \; \subst x {t_2}
{t_1}$$ in some positions within the term (here $\subst x {t_2}
{t_1}$ denotes the usual, capture-avoiding substitution of $t_2$ for
$x$ in $t_1$).  In order to capture formally the reduction relation on
terms, we can use an explicit representation of these positions using
contexts. Thus,
$$\plug C {\tapp {(\tlam x {t_1})} {t_2}} \; \rightarrow_{\beta} \; \plug C
{\subst x {t_2} {t_1}}$$ defines one step of the full,
nondeterministic reduction relation in the lambda calculus in a succinct way. This semantic format is called \emph{reduction semantics} \cite{Felleisen-Hieb:TCS92}.

\begin{example}
A term $\tapp{\tapp II}{(\tapp II)}$, where $I := \tlam x x$, can be reduced in one step by full $\beta$-reduction in two ways: $\tapp{I}{(\tapp II)} \leftarrow_\beta \tapp{\tapp II}{(\tapp II)} \to_\beta \tapp{\tapp II}{I}$
because both subterms $\tapp I I$ are contractible: $\tapp I I  \rightharpoonup_\beta I$, and both contexts $\tapp{\hole}{(\tapp II)}$ and $\tapp{\tapp II}{\hole}$ are correctly derived from nonterminal $C$ w.r.t. to the given grammar.
At the same time, the whole term is not contractible: $\tapp{\tapp II}{(\tapp II)} \not{\phantom{.}}\hspace{-3.2mm}\rightharpoonup_{\beta} \tapp{I}{(\tapp II)}$ because it is not an application of an abstraction.
\end{example}

In the following, the reflexive-transitive closure of any one-step
reduction relation $\rightarrow{}{}$ is denoted by $\twoheadrightarrow
{}{}$
(possibly with some decorations on arrows) and the
reflexive-symmetric-transitive closure is denoted by $=$, and is
called \textit{conversion}.

\begin{example}
The term $\tapp{\tapp II}{(\tapp II)}$ reduces (in three steps) to identity:
$\tapp{\tapp II}{(\tapp II)} \twoheadrightarrow_{\beta} I$ because
$\tapp{\tapp II}{(\tapp II)} \to_\beta \tapp{\tapp II}{I} \to_\beta \tapp{I}{I} \to_\beta I$.
The first reduct $\tapp{\tapp II}{I}$ does not reduce to the alternative reduct $\tapp{I}{(\tapp II)}$,
i.e., $\tapp{\tapp II}{I} \not{\phantom{.}}\hspace{-3.2mm}\twoheadrightarrow_{\beta} \tapp{I}{(\tapp II)}$,
but they are $\beta$-convertible:  $\tapp{\tapp II}{I} =_{\beta} \tapp{I}{(\tapp II)}$
because there exists a path of $\beta$-reductions and $\beta$-expansions between them:
$\tapp{\tapp II}{I} \leftarrow_\beta \tapp{\tapp II}{(\tapp II)} \to_\beta \tapp{I}{(\tapp II)}$.
\end{example}

Juxtaposition of two relations denotes
their composition, \eg, $s \, \twoheadrightarrow_{\beta} =_\alpha\, t$
means that $\exists t'.\; s\,\twoheadrightarrow_{\beta} \,t'\,
=_\alpha \,t$.

\begin{example}\label{ex:alpha}
Thanks to the usage of capture-avoiding substitution, all $\alpha$-renamings
needed to perform a $\beta$-contraction are implicitly performed under the
definition of \text{$\beta$-contraction}. For example,
$\tapp {(\tlam x {\tlam y {\tapp x y}})} y \rightharpoonup_\beta {\tlam z {\tapp y z}}$.
However, technically,
$\tlam x {\tapp {\tapp I I} x} \not{\phantom{.}}\hspace{-3.2mm}\twoheadrightarrow_{\beta} \tlam z {\tapp I z}$
because the renaming of $x$ to $z$ is outside of the capture-avoiding substitution.
Information about potential extra $\alpha$-renamings can be added on any side of the $\beta$-reduction:
\text{$\tlam x {\tapp {\tapp I I} x} \twoheadrightarrow_{\beta} =_\alpha\ \tlam z {\tapp I z}$}.
This way we can easily talk about terms up to renaming of bound variables.
\end{example}

If we want to impose a specific, deterministic reduction strategy, we
can write a more precise grammar of reduction contexts which narrows
down the positions where a computation step is allowed.  For example,
the call-by-name strategy, which is also known as a weak-head
reduction that applies $\beta$-contraction only in leftmost-outermost
positions, can be specified by the following definition of contexts:

$$ E ::= \hole \alt \tapp E t$$ and the
corresponding reduction relation is as follows:
$$\plug E {\tapp {(\tlam x {t_1})} {t_2}} \;
\stackrel{\cbn}{\to} \; \plug E {\subst x {t_2}
  {t_1}}$$
  
\begin{example}
The term $\tapp{\tapp II}{(\tapp II)}$ reduces in one step in call by name to $\tapp{I}{(\tapp II)}$, i.e., $\tapp{\tapp II}{(\tapp II)} \stackrel{\cbn}{\to} \tapp{I}{(\tapp II)}$ because the context $\tapp{\hole}{(\tapp II)}$ is a call-by-name context, i.e., it is derived from the nonterminal $E$,
while the context $\tapp{\tapp II}{\hole}$ is not because
it cannot be derived from $E$ and thus $\tapp{\tapp II}{(\tapp II)}$ does not reduce in call by name to $\tapp{\tapp II}{I}$, i.e., $\tapp{\tapp II}{(\tapp II)} \not{\phantom{.}}\hspace{-3.2mm}\stackrel{\cbn}{\to} \tapp{\tapp II}{I}$.
\end{example}
  
The generalization of call-by-name strategy to strong reduction does not stop on an
abstraction, but instead iterates the same strategy inside its body
and in arguments to neutral terms. It can be defined with contexts $N$
as follows:

\begin{alignat*}{3}
\NO \ni N &::= \overline{N} \alt \tlam x N\\
\overline{N} &::= \hole \alt \tapp {\overline{N}} t \alt \tapp a N
\end{alignat*}
where $a$ stands for neutral terms and $\nf$ are normal forms:
\begin{alignat*}{7}
\neu &::= & \tapp \neu \nf &\alt x\\
\nf &::= & \; \tlam x \nf &\alt \neu
\end{alignat*}

The definition of normal terms with the grammar (the nonterminal $n$) coincides with the definition of normal terms as terms that have no contractible subexpressions, i.e., normal forms of $\beta$-reduction, i.e., terms $t$ such that $t \not{\phantom{.}}\hspace{-3.2mm}\to_{\beta}$.

\begin{example}
The identity $I$ is a normal term (i.e.,
$I \not{\phantom{.}}\hspace{-3.2mm}\to_{\beta}$), while $\tapp I I$ is not
(i.e., $\tapp I I \to_{\beta} I$).
Similarly, $\tlam y I$ is a normal term (i.e.,
$\tlam y I \not{\phantom{.}}\hspace{-3.2mm}\to_{\beta}$), while
$\tlam y {(\tapp I I)}$ is not a normal term:
(i.e., $\tlam y {(\tapp I I)} \to_{\beta} \tlam y I$).
Moreover, $\tapp y I$ is a neutral term, while $\tapp y {(\tapp I I)}$ is not.
Every neutral term is a normal term that is not an abstraction.
\end{example}

This strong strategy is known as the normal-order strategy, and it
computes full normal forms $\nf$ of lambda terms.
We define one step of normal-order reduction as follows:
$$ \plug {N} {\tapp {(\tlam x {t_1})} {t_2}} \;
\stackrel{\no}{\to} \; \plug {N} {\subst x {t_2}
  {t_1}}$$

The grammar of contexts $N$ as defined above builds them from the
outside in. It is often convenient to build contexts inside-out and
think about them as stacks of single frames, where the top of the
stack is the frame surrounding the hole:

\begin{alignat*}{3}
\NO \ni {N} &::= \underline N \alt {N}[\tapp \hole t]\\
\underline N &::= \hole \alt \underline N[\tlam x \hole] \alt {N}[\tapp \neu \hole]
\end{alignat*}

This inside-out grammar makes it explicit how contexts are constructed
and deconstructed in abstract machines (and we directly use it in
Lemma \ref{lem:stack} later). \citet{Garcia-Nogueira:SCP14} presented a
translation from regular grammars of contexts to nondeterministic
finite automata, which was later used to transform inside-out and
outside-in grammars into each other
\citep{BiernackaBCD20}. Figure~\ref{fig:automata} shows the two
automata corresponding to the two grammars above. The difference
between them is that the directions of transitions are reversed and
the initial and final states are swapped. Therefore the two automata
recognize the same set of stacks, they just read them in the reversed
order: the inside-out automaton reads a stack starting from the frame
representing the hole of the context and moves towards the topmost
symbol, while the outside-in automaton starts from the topmost frame
and moves towards the hole. In consequence, both grammars (with $N$ as
the starting non-terminal) generate the same set of contexts.

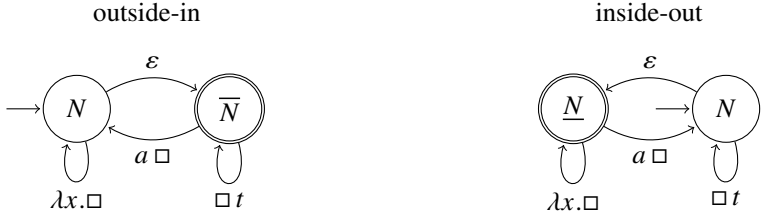
\begin{figure}[htb]
\begin{multicols}{2}
\begin{center}
outside-in\\
\vspace{2ex}
\begin{tikzpicture}[shorten >=1pt,node distance=2cm,on grid,auto, initial text={}]
\node[state,initial]   (N)              {$N$};
\node[state,accepting] (R) [right of=N] {$\overline{N}$};

\path[->]
(N) edge [loop below] node {$\tlam x \hole$}    ( )
    edge [bend left]  node {$\varepsilon$}      (R)
(R) edge [loop below] node {$\tapp \hole t$}    ( )
    edge [bend left]  node {$\tapp \neu \hole$} (N);
\end{tikzpicture}\phantom{$\longrightarrow$}
\end{center}

\columnbreak

\begin{center}
inside-out\\
\vspace{2ex}
\begin{tikzpicture}[shorten >=1pt,node distance=2cm,on grid,auto, initial text={}]
\node[state,accepting]   (F)              {$\underline{N}$};
\node[state,initial]     (N) [right of=F] {$N$};

\path[->]
(F) edge [loop below]         node {$\tlam x \hole$}    ( )
    edge [bend right, below]  node {$\tapp \neu \hole$} (N)
(N) edge [loop below]         node {$\tapp \hole t$}    ( )
    edge [bend right, above]  node {$\varepsilon$}      (F);
\end{tikzpicture}
\end{center}
\end{multicols}
\caption{Automata for the grammars of contexts}
\label{fig:automata}
\end{figure}

\begin{example}
The term $\tapp I I$ can be reduced in call by name: $\tapp I I \stackrel{\cbn}{\to} I$ because the empty context $\hole$ is a valid call-by-name context.
It can be also reduced in normal order: $\tapp I I \stackrel{\no}{\to} I$ because every call-by-name context is a normal-order context.
However, neither $\tlam y {\tapp I I}$ nor $\tapp y {(\tapp I I)}$ can be reduced in call by name:
$\tlam y {\tapp I I} \not{\phantom{...}}\hspace{-5mm}\stackrel{\cbn}{\to}$ and $\tapp y {(\tapp I I)} \not{\phantom{...}}\hspace{-5mm}\stackrel{\cbn}{\to}$
because $\tlam y \hole$ and $\tapp y \hole$ are not call-by-name contexts.
Nonetheless, these contexts are normal-order contexts (derivable from the nonterminal $N$) so both terms are normalized by normal order: $\tlam y {\tapp I I} \stackrel{\no}{\to} \tlam y I \not{\phantom{.}}\hspace{-3.2mm}\to_{\beta}$ and $\tapp y {(\tapp I I)} \stackrel{\no}{\to} \tapp y I \not{\phantom{.}}\hspace{-3.2mm}\to_{\beta}$.
\end{example}

The grammar of normal-order contexts restricts them in a way that the
reduction may take place only in the leftmost-outermost position.  For
example, in a term $\tlam y {\tapp {(\tlam z {\tapp I I})} I}$, normal
order contracts the subterm $\tapp {(\tlam z {\tapp I I})} I$ but
cannot contract the subterm $\tapp I I$, i.e.,
$ \tlam y {\tapp {(\tlam z I)} I}\;
\not{\phantom{.}}\hspace{-3.7mm}\stackrel{\no}{\leftarrow} \tlam y
{\tapp {(\tlam z {\tapp I I})} I} \stackrel{\no}{\to} \tlam y {\tapp I
  I}$ because
$\tlam y {\tapp {(\tlam z \hole)} I} = \hole[\tlam y \hole][\tapp
\hole I][\tlam z \hole]$ is not a normal-order context.  Contexts
derivable from the nonterminal $\underline{N}$ can be seen as
non-applicative (cf. \cite{Accattoli-DalLago:LMCS16}) normal-order
contexts, while derivable from $\overline{N}$ as normal-order contexts
not starting with an abstraction.  Many other strategies defined via
reduction semantics and their properties can be found in
\cite{Biernacka-al:ITP22}.

\subsection{The call-by-need strategy}
\label{sec:cbnd_strategy}

The ``lazy'' variant of the normal-order strategy (or of the weak
call-by-name strategy) needs to use some form of sharing.
To our knowledge, it has not been expressed with a simple restriction of
general contexts with the $\beta$-contraction without any extension of the
syntax of lambda terms.
It seems that expression of call by need by the grammar of contexts
cannot be done easily in the pure $\lambda$-calculus because of the tension
between the necessity to substitute the needy head variable and
the necessity to postpone substitution of the remaining ones.

Therefore, in order to provide a reduction semantics for it, some
syntax extension is often used.  We can either use a language with
explicit
substitutions~\cite{Balabonski-al:ICFP17,Balabonski-al:FSCD21} or
let-constructs~\cite{Ariola-al:POPL95} to explicitly handle bindings
of terms to variables, or simulate a store structure within the
reduction rules~\cite{Biernacka-Danvy:TCS07}.  Another known
possibility is to use a nonstandard contraction rule with a complex
grammar of context \cite{DBLP:conf/esop/ChangF12}.  In this paper, we
use an operational semantics for the weak call-by-need strategy in the
form of an abstract machine with explicit store
from~\cite{Danvy-Zerny:PPDP13}, presented in
Section~\ref{sec:conservative}.

\section{Higher-order normalizers}
\label{sec:derivation}
\subsection{Normal-order normalizer}
\label{sec:normalizer}

In Listing~\ref{fig:no-normalizer} we show a Racket program that
performs normal-order (strong call-by-name) normalization of
closed lambda terms into full normal form, and
is an instance of a normalization-by-evaluation algorithm. This
program is a slightly modified version of an OCaml NbE normalizer
reconstructed from Cr\'egut's abstract machine KN
in~\citep{BiernackaBCD20}, and similar to the
normalizer of \citet{Filinski-Rohde:RAIRO05}.

\begin{lstlisting}[float,caption={A higher-order normal-order normalizer},label={fig:no-normalizer}, numbers=left, numbersep=5pt, xleftmargin=10pt]
(struct *var (x  ) #:transparent)
(struct *app (l r) #:transparent)
(struct *lam (x t) #:transparent)
(struct  Abs (x t))

(define (reify v)
  (match v
    [(Abs x f)  (let ([x1 (gensym x)])
                  (*lam x1 (reify (f (λ () (*var x1))))))]
    [t          t]))

(define (apply-value v w)
  (match v
    [(Abs _ f)  (f w)]
    [t          (*app t (reify (w)))]))

(define (eval e t)
  (match t
    [(*var x)   ((hash-ref e x))]
    [(*app t u) (apply-value (eval e t) (λ () (eval e u)))] 
    [(*lam x t) (Abs x (λ (v) (eval (hash-set e x v) t)))]))

(define (normalize t)
  (reify (eval (hasheq) t)))
\end{lstlisting}
The idea of normalization by evaluation is to evaluate expressions of
an object language (here, lambda terms) into a semantic domain of
\emph{values} in such a way that equivalent terms are mapped to the
same value, and from the values it is possible to extract a syntactic
representation of the normal form (of the equivalence class).  Here,
our metalanguage is Racket and the domain of \emph{values} consists of
neutral terms and meta-level functions corresponding to lambda
abstractions.  Intuitively, meta-level functions should take values
and return values. However, Racket is an eager language, and in
normal-order strategy we do not want to evaluate an argument before
the evaluation of the caller function, so the evaluation of arguments
must be delayed. Therefore, meta-level functions operate on {\em
  thunks} (\ie, delayed values) and return values, as
in~\citep{Filinski-Rohde:RAIRO05}.

The algorithm is based on three functions: {\tt eval} that translates
a~lambda term to a~semantic value and calls its evaluation;
\texttt{apply-value} that says how to evaluate an application of a
meta-level function {\tt v} to a thunk {\tt w}; and \texttt{reify}
that translates semantic values to lambda terms in normal form.

Names {\tt *var}, {\tt *app}, {\tt *lam} are just constructors of
lambda terms. Environments are dictionaries that associate names of
free variables with thunks.  In Racket, dictionaries are accessed by
{\tt hash-ref} and updated by {\tt hash-set} functions.

To increase the readability of the final result, semantic functions
are coupled with the original variable name (\texttt{x} in line 21) by
the constructor {\tt Abs} and later used in the generation of fresh
variables in function \texttt{reify} in line 8; since name generation
is called inside reification, all bound variables in the final result
are different.  The $\alpha$-renaming in full reduction is obligatory
because otherwise the term
$\tapp {(\tlam z {\tapp z z})} {(\tlam x {\tlam y {\tapp x y}})}$
would incorrectly reduce to $ {(\tlam y {\tlam y {\tapp y y}})}$.

To normalize a closed term it is enough to evaluate it in the empty
environment ({\tt hasheq}) and to reify its value---this is
done in line 24.

\subsection{An optimized, call-by-need normalizer}
\label{sec:cbndNormalizer}
\begin{lstlisting}[float,caption={A call-by-need normalizer},label={fig:cbnd-normalizer}, numbers=left, numbersep=5pt, xleftmargin=10pt]
(struct todo (t))
(struct done (v) #:transparent)

(define (force mt)
  (match (unbox mt)
    [(done v) v]
    [(todo t) (let ([v (t)]) (set-box! mt (done v)) v)]))

(define (memothunk thunk)
  (box (todo thunk)))

(define (reify v)
  (match v
    [(Abs c _)  (force c)]
    [t          t]))

(define (apply-value v w)
  (match v
    [(Abs _ f)  (f w)]
    [t          (*app t (reify (w)))]))

(define (eval e t)
  (match t
    [(*var x)   (force (hash-ref e x (box (done (*var x)))))]
    [(*app t u) (apply-value (eval e t) (λ () (eval e u)))] 
    [(*lam x t) (let ([f (λ (v) (eval (hash-set e x (memothunk v)) t))])
                 (Abs (memothunk
                        (λ () (let ([x1 (gensym x)])
                                (*lam x1 (reify (f (λ () (*var x1))))))))
                       f))]))

(define (normalize t)
  (reify (eval (hasheq) t)))
\end{lstlisting}

We optimize the normalizer for normal-order by introducing the
standard technique of memothunks to cache computed values and reuse
them later, if needed. We also extend it to handle open
terms.
The resulting normalizer is shown in
Listing~\ref{fig:cbnd-normalizer}.

We introduce two more constructors to explicitly indicate the state of
a~memothunk: {\tt todo} stores plain thunks to be computed and {\tt
  done} stores their already computed results.  A memothunk consists of
Racket's built-in mutable memory cell called {\tt box} storing one of
these two.  The function {\tt memothunk} converts a plain thunk into a
memothunk by putting it into {\tt todo} and {\tt box}.  To access a
memothunk we use the procedure {\tt force} that just returns the
result of a computation, if it is already computed. Otherwise, the procedure
{\tt force} runs the stored plain thunk, remembers its result in the
{\tt box}, and then returns the result.

There are several changes in the normalizer, let us start with less
important ones. Recall that normalization consists of two phases:
evaluation in the semantic domain followed by reification in the
syntactic domain (\ie, in the object language).  Therefore things that
happen on the borderline between the two phases may be called either
from \texttt{eval} or from \texttt{reify} function. In particular,
generation of fresh names is moved from line 8 in
Listing~\ref{fig:no-normalizer} to line~28 in
Listing~\ref{fig:cbnd-normalizer}.
In consequence, while in the
 normal-order normalizer a fresh name is generated for each reification
of each abstraction,
in the lazy normalizer reification of each abstraction generates at
most one fresh name.  However, this is fine because normal forms
retrieved from the same memothunk do not overlap (cf.\ the two binding
occurrences of the variable $z_0$ in step 27 in
Table~\ref{fig:execution_ex_ref}).
Similarly, the whole code of reification from lines 8--9 in Listing~\ref{fig:no-normalizer}
is moved as a thunk to lines~28--29 in Listing~\ref{fig:cbnd-normalizer}.

A~more important change from the efficiency point of view is that the semantic function constructed in
line~26 converts a plain thunk that it receives into a memothunk when
it adds it into the environment.  By creating a memothunk for the
argument, we thus avoid its recomputation in the semantic domain.
Moreover, in line~27 we create a memothunk for the second phase, which
avoids recomputation of the reification in the syntactic domain. Now
the \texttt{Abs} constructor carries two parameters: the memothunk
used for computation of normal forms in the syntactic domain in
line~14 and the semantic function used for evaluation in the semantic
domain in line~19.
The two types of memothunks are then {\tt forced} in line~24 (for
evaluation in the semantic domain), and in line~14 (for normalization
in the syntactic domain).

There is a design choice in putting a \texttt{memothunk} in line 26 rather than in line 25, where a plain thunk is passed as an argument.
Both would be correct, but the one presented is a bit more sparing.
If the left term \texttt{t} in line 25 evaluates to an \texttt{Abs}, the value of right term \texttt{u} is put into \texttt{memothunk} immediately anyway.
However, if the left term evaluates to a neutral term,
we know that the right term's value will be reified exactly once, so we can
avoid memoization overhead.

To be able to handle open input terms, we also add a default
value for {\tt hash-ref} in line~24: when the variable \texttt{x} does
not occur in the environment \texttt{e}, a box with the value
\texttt{x} is returned. 

\subsection{Derivation of an abstract machine}\label{sec:derivation-machine}
\begin{lstlisting}[float,caption={Input specification to {\tt semt}},label={fig:semt-input}]
(def-data Term
  {Var String}
  {App Term Term}
  {Lam String Term})
(def-struct {Abs Any Any})
(def-struct {Memo Any})

(def gensym #:atomic (x) x) ;; dummy
(def memothunk (t) {Memo t})
(def force (m) (match m ({Memo t} (t)))) ;; dummy

(def init #:atomic (x) (memothunk (fun () {Var x})))
(def extend #:atomic (env k v)
  (fun #:atomic #:no-defun (x)
       (match (eq? x k)
         (#t v)
         (#f (env x)))))

(def reify (v)
  (match v
    ({Abs c f} (force c))
    (t         t)))

(def apply-value (v w)
  (match v
    ({Abs _ f} (f w))
    (t         {App t (reify (w))})))

(def eval (env [Term term])
  (match term
    ({Var x}      (force (env x)))
    ({App fn arg} (apply-value (eval env fn)
                               (fun #:name Closure () (eval env arg))))
    ({Lam x body} (let f (fun #:name AbsClos #:apply apply-closure (v)
                              (eval (extend env x (memothunk v)) body)))
                  {Abs (memothunk (fun () (let x1 (gensym x))
                                    {Lam x1 (reify (f (fun #:name VVar 
                                                           #:apply render 
                                                           () {Var x1})))}))
                       f})))

(def main ([Term term])
  (reify (eval init term)))

\end{lstlisting}
We used a semantic transformation tool {\tt semt} developed by
\citet{BuszkaB21:LOPSTR21} to automatically transform the initial
evaluator for Strong CbNd into an abstract machine.  In
Listing~\ref{fig:semt-input} we present the input code for {\tt semt}
written in IDL (Interpreter Definition Language), which differs from
the Racket evaluator slightly in syntax.  Constructors of lambda terms
are {\tt Var}, {\tt App}, and {\tt Lam}.  The input code contains a
handful of annotations that guide the derivation process, and a
``dummy'' implementation of the functions {\tt gensym} and {\tt
  force}. The reason for this latter tweak is that the transformer
does not handle stateful computation, \ie, it does not support mutable
memory cells.  In order to get the correct implementation of laziness,
after we transform the evaluator into the machine, we then replace
back the implementation of the {\tt force} function with that of
Listing~\ref{fig:cbnd-normalizer}, and we use the default fresh-name
generating function.

The transformation tool implements a refinement of the known
derivation methodology of Danvy et al.'s {\em functional
  correspondence}~\cite{Ager-al:PPDP03}, and it performs the
following steps:
\begin{enumerate}
\item translation to administrative normal form (ANF), where all
  intermediate results of computation are bound by let-constructors, 
\item selective transformation into continuation-passing
  style (a.k.a. CPS transformation),
\item selective defunctionalization of higher-order function spaces into
  first-order data structures and the corresponding dispatch
  functions.
\end{enumerate}
The paper~\cite{Ager-al:PPDP03} presents the methodology explained on
several examples including the Krivine and CEK machines.

The selective CPS transformation does not alter functions annotated
with {\tt \#:atomic} which are treated as atomic, ``trivial''
operations. In our case these are the functions operating on
environments, and the {\tt gensym} function that generates fresh
names. If we dropped the annotation {\tt \#:atomic} for the dummy
implementation of {\tt gensym}, the transformation could simplify it
as an identity. Generally, we used these as the boundaries of the purely
functional core of the normalizer. After the transformation, we could replace
the dummy implementation of {\tt gensym} with the real one.

In an extreme case, if the {\tt normalize} function was treated as atomic,
we would get a machine with a single transition that computes the normal form
of a term in one step, but it would not be a useful artefact for the complexity
analysis.

In the final stage of the transformation, defunctionalization is
performed on all higher-order functions except those that have been
annotated with {\tt \#:non-defun}, which are left untouched. In our
case, again, we do not touch the function {\tt extend} that operates
on environments.
Similarly, if we dropped the annotations for functions representing
environments, we would obtain a machine which traverses the environments step
by step.

The annotations {\tt \#:name} and {\tt \#:apply}
(each with a single argument) are used to supply non-generic names for
the datatype and the apply-function constructed by
defunctionalization.

\section{The abstract machine}
\label{sec:machine}

In this section we present the machine obtained as the result of the
transformation. We show its definition translated into mathematical
notation so that it can be easily read and understood without the need
to resort to Racket syntax.

The machine is presented in Table~\ref{fig:machine}. It works with
pure lambda terms, local environments, a stack and a global
store. Environments assign locations on the store to variables. The
syntactic category of values contains terms (intuitively, lambda terms
in normal form) and intermediate values (intuitively, abstractions)
being closures consisting of a lambda abstraction paired with a
local environment and annotated with a location on the global store. The
store assigns storable values, or memothunks, to locations.
Memothunks
are annotated with either~$\mbox{\textcolor{red!80!black}\XSolidBrush}$ (which stands for \emph{to do}
and indicates a code that needs to be executed) or~$\mbox{\textcolor{green!50!black}\Checkmark}$
(which stands for \emph{done} and indicates that a value is already
computed and memoized). A stack is simply a sequence of frames;
intuitively, it represents the evaluation context of the currently
evaluated term. There are two kinds of configurations corresponding to
two modes of operation: in configurations
$\langle c, s, \sigma \rangle_\etrian$ the machine evaluates the
closure $c$ to a weak normal form, and in $ \cconf \sigma v s$ it
continues with already computed value $v$.

\begin{table}[htb!]
\caption{The RKNL abstract machine, a reasonable and lazy variant of KN}
\label{fig:machine}
\begin{alignat*}{3}
\mathit{Identifiers} \ni &\;& x &\\
\mathit{Terms}       \ni && t   &::= x \alt \tapp{t_1}{t_2} \alt \tlam x t\\
\mathit{Locations} \ni && \ell\\
\mathit{Envs}   \ni && e   & \;\subtype \fun {\mathit{Identifiers}} {\mathit{Locations}} \\
\mathit{Closures} \ni && c   &::= \clos t e\\
\mathit{Values} \ni && v   &::= t \alt \abs {\tlam x t} e \ell\\
\mathit{Storable\;Values} \ni && \_ & ::= \tobedone \bot \alt \tobedone c \alt \done v\\
\mathit{Stores} \ni && \;\sigma & \;\subtype \fun {\mathit{Locations}} {\mathit{Storable\;Values}}\\
\mathit{Frames} \ni && f   &::= \frapp c
    \alt \flapp t
    \alt \flam x
    \alt \fcache \ell \\
\mathit{Stacks} \ni && s   &::= \nil \alt \cons{f}{s} \\
\mathit{Confs} \ni && k   &::=
  \langle c, s, \sigma \rangle_\etrian
  \alt \cconf \sigma v s\\
\mathit{Transitions:}\\[-3em]
\end{alignat*}%

\begin{align}
   t &\mapsto \econf \nil t \nil \nil \nonumber\\
\econf \sigma {\tapp{t_1}{t_2}} e {\hspace{16.5mm}s} &\to \econf         \sigma     {t_1} e {\cons{\frclos{t_2} e} {s}} \tag{1}\label{tr:1}\\
\econf \sigma       {\tlam x t} e {\hspace{16.5mm}s} &\to \cconf {\alloc \sigma \ell {\tobedone \bot}} {\abs {\tlam x t} e \ell} s  \tag{2}\label{tr:2}\\
  \econf \sigma              x    e {\hspace{16.5mm}s} &\to \econf \sigma t {e_2} {\cons { \fcache {\ell}} s} \; \nonumber\\
 & \hspace{2.45cm} \mbox{where } \ell={\lookup e x},\;{\lookup \sigma {\ell}} = \tobedone {\clos t {e_2}} \tag{3}\label{tr:3}\\
\econf \sigma              x    e {\hspace{16.5mm}s} &\to \cconf \sigma v s \hspace{5.5mm} \mbox{where } {\lookup \sigma {\lookup e x}} = \done v \vee (v = x \notin e) \tag{4}\label{tr:4}\\
\cconf \sigma v {\hspace{6.4mm}\cons {\fcache \ell} s} &\to \cconf {\update \sigma \ell {\done v}} v s \tag{5}\label{tr:5}\\
\cconf \sigma {\abs {\tlam x t} e {\ell}} {\cons {\frclos{t_2} {e_2}} s} &\to \econf {\alloc \sigma {\ell_2} {\tobedone {\clos{t_2} {e_2}}}} t {\update e x {\ell_2}} s \tag{6}\label{tr:6}\\
\cconf \sigma {\abs {\tlam x t} e {\ell}} {\hspace{16.5mm}s} &\to \econf {\alloc \sigma {\ell_2} {\done {\fresh x}}} t {\update e x {\ell_2}} {\cons {\flam {\fresh x}} {\cons {\fcache \ell} s}} \nonumber\\
& \hspace{3.848cm}\mbox{where } \lookup \sigma \ell = \tobedone \bot \tag{7}\label{tr:7} \\
\cconf \sigma {\abs {\tlam x t} e \ell} {\hspace{16.5mm}s} &\to \cconf \sigma v s \hspace{19.5mm}\mbox{where } \lookup \sigma \ell = \done v \tag{8}\label{tr:8}\\ 
\cconf \sigma t                         {\cons {\frclos{t_2} {e_2}} s} &\to \econf \sigma {t_2} {e_2} {\cons {\flapp t} s}\tag{9}\label{tr:9}\\
\cconf \sigma {t_2} {\hspace{7.6mm}\cons {\flapp {t_1}} s} &\to \cconf \sigma {\tapp {t_1} {t_2}} s\tag{10} \label{tr:10}\\ 
\cconf \sigma t {\hspace{6.5mm}\cons {\flam x} s} &\to \cconf \sigma {\tlam x t} s \tag{11}\label{tr:11}\\
\cconf \sigma t {\hspace{14.7mm}\nil} &\mapsto t \nonumber
\end{align}
\end{table}

The machine starts by loading the input term to the initial
configuration with empty environment, empty stack and empty store.
Then it proceeds through successive transitions.  The first transition
that matches the current configuration should be applied, so
Table~\ref{fig:machine} describes a deterministic abstract machine.
The first six transitions are standard, they directly correspond to
transitions of the lazy variant of the Krivine machine from
~\cite{Danvy-Zerny:PPDP13} (we discuss this correspondence in
Section~\ref{sec:conservative}). To evaluate application
$\tapp{t_1}{t_2}$, transition ~(\ref{tr:1}) calls the evaluation of
$t_1$ and pushes a closure pairing $t_2$ with the current environment
to the stack. In the case of a lambda abstraction $\tlam x t$,
transition~(\ref{tr:2}) primarily changes the mode of operation. A
non-standard part, added for memoization, is that ~(\ref{tr:2})
allocates a fresh location $\ell$ on the store, and fills it with a
placeholder for the (strong) normal form of $\tlam x t$.
Transitions~(\ref{tr:3}) and~(\ref{tr:4}) apply when the value of the
formal parameter $x$ is actually needed.  In this case, the actual
parameter is expected to be stored in a thunk at location $e(x)$
indicated by the environment $e$. If the thunk contains code to be
executed, the location is pushed to the stack and the code is forced
by transition~(\ref{tr:3}); otherwise, the thunk contains a memoized
value which is simply read by transition~(\ref{tr:4}).  If the
environment $e$ does not have an entry for $x$, then $x$ turns out to
be an open variable and it is taken as a value (this is the default
value mentioned at the end of Section~\ref{sec:cbndNormalizer}).  The
disjunctive condition of transition~(\ref{tr:4}) is an artefact of the
design choice made already in the normalizer:
In~Listing~\ref{fig:cbnd-normalizer} (line 24), in the variable case of
the \texttt{eval} function, the variable is the default value for the
environment lookup, for the case when it is not found in the
environment.  Transition~(\ref{tr:5}) implements memoization of actual
parameters: when evaluation of an actual parameter is finished, the
top of the stack contains the location assigned to the formal
parameter, and this location is updated with the computed
value. Transition~(\ref{tr:6}) implements $\beta$-contraction and
delays the evaluation of the actual parameter: new location $\ell_2$
is created, a thunk with the actual parameter is stored at $\ell_2$,
and the evaluation of the body of the lambda abstraction (with an
appropriately updated environment) is called.

Transitions~(\ref{tr:7}) and~(\ref{tr:8}) implement normalization of
lambda abstractions. They apply when the currently processed closure
is a lambda abstraction (annotated with a location $\ell$) and there
is no argument on the stack (so that transition~(\ref{tr:6}) does not
apply). If the thunk at $\ell$ contains a memoized value,
transition~(\ref{tr:8}) simply reads this value. Otherwise, it contains
the placeholder $\tobedone \bot$, and transition~(\ref{tr:7}) starts
the normalization of the abstraction: the location $\ell$ is pushed to
the stack; the variable $x$ is $\alpha$-renamed to a fresh variable
$\fresh x$ and stored in a fresh location $\ell_2$; the environment is
acknowledged about $\alpha$-renaming; a context $\flam {\fresh x}$ is
pushed to the stack; the evaluation of the body of the abstraction is
called. When this evaluation returns, transition~(\ref{tr:11})
reconstructs the normalized lambda abstraction and then
transition~(\ref{tr:5}) memoizes it at location $\ell$.

Transition~(\ref{tr:9}) implements normalization of neutral
terms. Here $t$ is in normal form, but not an abstraction
(see Lemma~\ref{lem:shape}), so it is a
neutral term; there is an argument on the top of the stack, and this argument
is now evaluated. Transitions~(\ref{tr:10}) and~(\ref{tr:11})
reconstruct the final normal forms: transition~(\ref{tr:10}) deals
with neutral terms and~(\ref{tr:11}) with lambda abstractions. Finally,
the result is unloaded from a final configuration.

\begin{remark}\label{rem:persistence}
  The whole machine is intended to be naturally implementable as
  a~\emph{persistent} data structure (cf. \cite{Okasaki:99}). That
  means that old configurations of the machine are not destroyed by
  any transition, so they are accessible if the user has references to
  them.  This approach is imposed if data constructors are immutable, 
  which is default in most functional languages including Racket, 
  OCaml and Haskell.

  From the perspective of the higher-order normalizer, it applies even to the
  store, employed by memothunks. Memothunks are what Okasaki calls suspensions
  and denotes using \text{\$-notation} in his book. When a memothunk is forced,
  the retrieved value is the same, no matter if it was already computed or not.
  In this sense, references to memothunks are persistent.

  This flavour of persistence is translated to the machine implemented with
  a mutable store, so it uses mutable references very rigorously.
  Each mutable reference is overwritten at most
  once, from \emph {to do} to \emph{done}. When referencing an old
  configuration, it may happen that some values of memothunks are already
  fast-forwarded, so the machine will perform smaller number of transitions
  than originally.
  Alternatively, at the cost of logarithmic overhead,
  the store can be implemented as a persistent dictionary.

  Persistent data structures allow free sharing of substructures:
  since a substructure is never destroyed in the processing, it is
  enough to share a reference. We assume a shared representation of
  terms as in~\cite{CondoluciAC19}, as well as of stacks and
  environments, without introduction of explicit substitutions or
  let-constructs. In particular, in transition~(\ref{tr:1}) it is the
  reference to an environment $e$ (and not the environment itself)
  that is copied, so the transition is performed in constant time.
  Similarly, in transition~(\ref{tr:3}) only a~reference to a~term $t$
  is copied.

  For example, exponentially big normal forms of the family
  $e_n = \tlam{x}{\tapp{\tapp{c_n}{\omega}}{x}}$, where
  $ \omega := \tlam{x}{\tapp{x}{x}}$ and $c_n$ denotes the
  $n^\text{th}$ Church numeral, consume only a~linear in $n$ amount of
  memory and are computed in linear time.
\end{remark}

\subsection{Elaborate example execution}

We present the behaviour of the RKNL machine in the elaborate example
execution in Table~\ref{fig:execution_ex_ref}.  It normalizes the term
${\tapp {({\tlam {x} {{\tapp {{\tapp {c} {x}}} {x}}}})} {\big({\tapp
      {({\tlam {y} {{\tlam {z} {{\tapp {I} {z}}}}}})}
      {\Omega}}\big)}}$ where $I := \tlam x x$ is the identity,
$\Omega := \omega\omega$ is a well-known divergent term (with
$\omega := \tlam x {\tapp x x}$), and $c$ is a free variable.
This is one of the simplest
examples that uses all transitions of the machine and demonstrates its
main features: the machine is able to evaluate open terms (the
variable $c$ is free) and terms under $\lambda$-abstractions (the
subterm $\tlam {z} {{\tapp {I} {z}}}$ is reduced to
${\tlam {z} {z}}$); not-needed arguments are not evaluated (the
variable $y$ is not used in
${\tlam {y} {{\tlam {z} {{\tapp {I} {z}}}}}}$, so $\Omega$ is not
evaluated); needed arguments are evaluated only once (even if $x$ is
used twice in ${\tlam {x} {{\tapp {{\tapp {c} {x}}} {x}}}}$, the
actual argument
${\tapp {({\tlam {y} {{\tlam {z} {{\tapp {I} {z}}}}}})} {\Omega}}$ is
evaluated only once).
\begin{table}[p]
\caption{Elaborate example execution in refocusing notation}
\label{fig:execution_ex_ref}

{
\begin{align*}
0:&& {\cut {{\clos {{\tapp {({\tlam {x} {{\tapp {{\tapp {c} {x}}} {x}}}})} {({\tapp {A} {\Omega}})}}} {\![]}}}_\etrian} &| [] &&\stackrel{(\ref{tr:1})}{\to}\\[-0.6ex]
1:&& {\tapp {{\cut {{\clos {{\tlam {x} {{\tapp {{\tapp {c} {x}}} {x}}}}} {\![]}}}_\etrian}} {{\clos {{\tapp {A} {\Omega}}} {\![]}}}} &| [] &&\stackrel{(\ref{tr:2})}{\to}\\[-0.6ex]
2:&& {\tapp {{\cut {{{{\mathbbm{a}}}\!:=\!{{\clos {{\tlam {x} {{\tapp {{\tapp {c} {x}}} {x}}}}} {\![]}}}}}_\ctrian}} {{\clos {{\tapp {A} {\Omega}}} {\![]}}}} &| [{\mathbbm{a}}\!\mapsto\!{\tobedone {\bot}}] &&\stackrel{(\ref{tr:6})}{\to}\\[-0.6ex]
3:&& {\cut {{\clos {{\tapp {{\tapp {c} {x}}} {x}}} {\![x\!\mapsto\!{\mathbbm{x}}]}}}_\etrian} &| [{\mathbbm{x}}\!\mapsto\!{\tobedone {{\clos {{\tapp {A} {\Omega}}} {\![]}}}}] &&\stackrel{(\ref{tr:1})}{\to}\\[-0.6ex]
4:&& {\tapp {{\cut {{\clos {{\tapp {c} {x}}} {\![x\!\mapsto\!{\mathbbm{x}}]}}}_\etrian}} {x^{\mathbbm{x}}}} &| [{\mathbbm{x}}\!\mapsto\!{\tobedone {{\clos {{\tapp {A} {\Omega}}} {\![]}}}}] &&\stackrel{(\ref{tr:1})}{\to}\\[-0.6ex]
5:&& {\tapp {{\tapp {{\cut {{\clos {c} {\![x\!\mapsto\!{\mathbbm{x}}]}}}_\etrian}} {x^{\mathbbm{x}}}}} {x^{\mathbbm{x}}}} &| [{\mathbbm{x}}\!\mapsto\!{\tobedone {{\clos {{\tapp {A} {\Omega}}} {\![]}}}}] &&\stackrel{(\ref{tr:4})}{\to}\\[-0.6ex]
6:&& {\tapp {{\tapp {{\cut {c}_\ctrian}} {x^{\mathbbm{x}}}}} {x^{\mathbbm{x}}}} &| [{\mathbbm{x}}\!\mapsto\!{\tobedone {{\clos {{\tapp {A} {\Omega}}} {\![]}}}}] &&\stackrel{(\ref{tr:9})}{\to}\\[-0.6ex]
7:&& {\tapp {{\tapp {c} {{\cut {x^{\mathbbm{x}}}_\etrian}}}} {x^{\mathbbm{x}}}} &| [{\mathbbm{x}}\!\mapsto\!{\tobedone {{\clos {{\tapp {A} {\Omega}}} {\![]}}}}] &&\stackrel{(\ref{tr:3})}{\to}\\[-0.6ex]
8:&& {\tapp {{\tapp {c} {\left({{{\mathbbm{x}}}\!:=\!{{\cut {{\clos {{\tapp {A} {\Omega}}} {\![]}}}_\etrian}}}\right)}}} {x^{\mathbbm{x}}}} &| [{\mathbbm{x}}\!\mapsto\!{\tobedone {{\clos {{\tapp {A} {\Omega}}} {\![]}}}}] &&\stackrel{(\ref{tr:1})}{\to}\\[-0.6ex]
9:&& {\tapp {{\tapp {c} {\left({{{\mathbbm{x}}}\!:=\!\left({{\tapp {{\cut {{\clos {A} {\![]}}}_\etrian}} {{\clos {\Omega} {\![]}}}}}\right)}\right)}}} {x^{\mathbbm{x}}}} &| [{\mathbbm{x}}\!\mapsto\!{\tobedone {{\clos {{\tapp {A} {\Omega}}} {\![]}}}}] &&\stackrel{(\ref{tr:2})}{\to}\\[-0.6ex]
10:&& {\tapp {{\tapp {c} {\left({{{\mathbbm{x}}}\!:=\!\left({{\tapp {{\cut {{{{\mathbbm{b}}}\!:=\!{{\clos {A} {\![]}}}}}_\ctrian}} {{\clos {\Omega} {\![]}}}}}\right)}\right)}}} {x^{\mathbbm{x}}}} &| [{\mathbbm{x}}\!\mapsto\!{\tobedone {{\clos {{\tapp {A} {\Omega}}} {\![]}}}}, {\mathbbm{b}}\!\mapsto\!{\tobedone {\bot}}] &&\stackrel{(\ref{tr:6})}{\to}\\[-0.6ex]
11:&& {\tapp {{\tapp {c} {\left({{{\mathbbm{x}}}\!:=\!{{\cut {{\clos {{\tlam {z} {{\tapp {I} {z}}}}} {\![y\!\mapsto\!{\mathbbm{y}}]}}}_\etrian}}}\right)}}} {x^{\mathbbm{x}}}} &| [{\mathbbm{x}}\!\mapsto\!{\tobedone {{\clos {{\tapp {A} {\Omega}}} {\![]}}}}, {\mathbbm{y}}\!\mapsto\!{\tobedone {{\clos {\Omega} {\![]}}}}] &&\stackrel{(\ref{tr:2})}{\to}\\[-0.6ex]
12:&& {\tapp {{\tapp {c} {\left({{{\mathbbm{x}}}\!:=\!{{\cut {{{{\mathbbm{d}}}\!:=\!{{\clos {{\tlam {z} {{\tapp {I} {z}}}}} {\![y\!\mapsto\!{\mathbbm{y}}]}}}}}_\ctrian}}}\right)}}} {x^{\mathbbm{x}}}} &| [{\mathbbm{x}}\!\mapsto\!{\tobedone {{\clos {{\tapp {A} {\Omega}}} {\![]}}}}, {\mathbbm{d}}\!\mapsto\!{\tobedone {\bot}}, {\mathbbm{y}}\!\mapsto\!{\tobedone {{\clos {\Omega} {\![]}}}}] &&\stackrel{(\ref{tr:5})}{\to}\\[-0.6ex]
13:&& {\tapp {{\tapp {c} {{\cut {{{{\mathbbm{d}}}\!:=\!{{\clos {{\tlam {z} {{\tapp {I} {z}}}}} {\![y\!\mapsto\!{\mathbbm{y}}]}}}}}_\ctrian}}}} {x^{\mathbbm{x}}}} &| \,\sigma_1 \!\ast\! [{\mathbbm{d}}\!\mapsto\!{\tobedone {\bot}}] &&\stackrel{(\ref{tr:7})}{\to}\\[-0.6ex]
14:&& {\tapp {{\tapp {c} {\left({{{\mathbbm{d}}}\!:=\!{{\tlam {{z_0}} {{\cut {{\clos {{\tapp {I} {z}}} {\!e_\mathit{yz}}}}_\etrian}}}}}\right)}}} {x^{\mathbbm{x}}}} &| \,\sigma_1 \!\ast\! [{\mathbbm{d}}\!\mapsto\!{\tobedone {\bot}}, {\mathbbm{z}}\!\mapsto\!{\done {{z_0}}}] &&\stackrel{(\ref{tr:1})}{\to}\\[-0.6ex]
15:&& {\tapp {{\tapp {c} {\left({{{\mathbbm{d}}}\!:=\!\left({{\tlam {{z_0}} {{\tapp {{\cut {{\clos {I} {\!e_\mathit{yz}}}}_\etrian}} {{\clos {z} {\!e_\mathit{yz}}}}}}}}\right)}\right)}}} {x^{\mathbbm{x}}}} &| \,\sigma_1 \!\ast\! [{\mathbbm{d}}\!\mapsto\!{\tobedone {\bot}}, {\mathbbm{z}}\!\mapsto\!{\done {{z_0}}}] &&\stackrel{(\ref{tr:2})}{\to}\\[-0.6ex]
16:&& {\tapp {{\tapp {c} {\left({{{\mathbbm{d}}}\!:=\!\left({{\tlam {{z_0}} {{\tapp {{\cut {{{{\mathbbm{e}}}\!:=\!{{\clos {I} {\!e_\mathit{yz}}}}}}_\ctrian}} {{\clos {z} {\!e_\mathit{yz}}}}}}}}\right)}\right)}}} {x^{\mathbbm{x}}}} &| \,\sigma_1 \!\ast\! [{\mathbbm{d}}\!\mapsto\!{\tobedone {\bot}}, {\mathbbm{e}}\!\mapsto\!{\tobedone {\bot}}, {\mathbbm{z}}\!\mapsto\!{\done {{z_0}}}] &&\stackrel{(\ref{tr:6})}{\to}\\[-0.6ex]
17:&& {\tapp {{\tapp {c} {\left({{{\mathbbm{d}}}\!:=\!{{\tlam {{z_0}} {{\cut {{\clos {x} {\!e_\mathit{yz}\!\ast\![x\!\mapsto\!{\mathbbm{w}}]}}}_\etrian}}}}}\right)}}} {x^{\mathbbm{x}}}} &| \,\sigma_1 \!\ast\! [{\mathbbm{d}}\!\mapsto\!{\tobedone {\bot}}, {\mathbbm{w}}\!\mapsto\!{\tobedone {{\clos {z} {\!e_\mathit{yz}}}}}, {\mathbbm{z}}\!\mapsto\!{\done {{z_0}}}] &&\stackrel{(\ref{tr:3})}{\to}\\[-0.6ex]
18:&& {\tapp {{\tapp {c} {\left({{{\mathbbm{d}}}\!:=\!{{\tlam {{z_0}} {{{{\mathbbm{w}}}\!:=\!{{\cut {{\clos {z} {\!e_\mathit{yz}}}}_\etrian}}}}}}}\right)}}} {x^{\mathbbm{x}}}} &| \,\sigma_1 \!\ast\! [{\mathbbm{d}}\!\mapsto\!{\tobedone {\bot}}, {\mathbbm{w}}\!\mapsto\!{\tobedone {{\clos {z} {\!e_\mathit{yz}}}}}, {\mathbbm{z}}\!\mapsto\!{\done {{z_0}}}] &&\stackrel{(\ref{tr:4})}{\to}\\[-0.6ex]
19:&& {\tapp {{\tapp {c} {\left({{{\mathbbm{d}}}\!:=\!{{\tlam {{z_0}} {{{{\mathbbm{w}}}\!:=\!{{\cut {{z_0}}_\ctrian}}}}}}}\right)}}} {x^{\mathbbm{x}}}} &| \,\sigma_1 \!\ast\! [{\mathbbm{d}}\!\mapsto\!{\tobedone {\bot}}, {\mathbbm{w}}\!\mapsto\!{\tobedone {{\clos {z} {\!e_\mathit{yz}}}}}] &&\stackrel{(\ref{tr:5})}{\to}\\[-0.6ex]
20:&& {\tapp {{\tapp {c} {\left({{{\mathbbm{d}}}\!:=\!{{\tlam {{z_0}} {{\cut {{z_0}}_\ctrian}}}}}\right)}}} {x^{\mathbbm{x}}}} &| \,\sigma_1 \!\ast\! [{\mathbbm{d}}\!\mapsto\!{\tobedone {\bot}}] &&\stackrel{(\ref{tr:11})}{\to}\\[-0.6ex]
21:&& {\tapp {{\tapp {c} {\left({{{\mathbbm{d}}}\!:=\!{{\cut {{\tlam {{z_0}} {{z_0}}}}_\ctrian}}}\right)}}} {x^{\mathbbm{x}}}} &| \,\sigma_1 \!\ast\! [{\mathbbm{d}}\!\mapsto\!{\tobedone {\bot}}] &&\stackrel{(\ref{tr:5})}{\to}\\[-0.6ex]
22:&& {\tapp {{\tapp {c} {{\cut {{\tlam {{z_0}} {{z_0}}}}_\ctrian}}}} {x^{\mathbbm{x}}}} &| \,\sigma_1 \!\ast\! [{\mathbbm{d}}\!\mapsto\!{\done {{\tlam {{z_0}} {{z_0}}}}}] &&\stackrel{(\ref{tr:10})}{\to}\\[-0.6ex]
23:&& {\tapp {{\cut {{\tapp {c} {{\tlam {{z_0}} {{z_0}}}}}}_\ctrian}} {x^{\mathbbm{x}}}} &| \,\sigma_1 \!\ast\! [{\mathbbm{d}}\!\mapsto\!{\done {{\tlam {{z_0}} {{z_0}}}}}] &&\stackrel{(\ref{tr:9})}{\to}\\[-0.6ex]
24:&& {\tapp {{\tapp {c} {({\tlam {{z_0}} {{z_0}}})}}} {{\cut {x^{\mathbbm{x}}}_\etrian}}} &| \,\sigma_1 \!\ast\! [{\mathbbm{d}}\!\mapsto\!{\done {{\tlam {{z_0}} {{z_0}}}}}] &&\stackrel{(\ref{tr:4})}{\to}\\[-0.6ex]
25:&& {\tapp {{\tapp {c} {({\tlam {{z_0}} {{z_0}}})}}} {{\cut {{{{\mathbbm{d}}}\!:=\!{{\clos {{\tlam {z} {{\tapp {I} {z}}}}} {\![y\!\mapsto\!{\mathbbm{y}}]}}}}}_\ctrian}}} &| [{\mathbbm{d}}\!\mapsto\!{\done {{\tlam {{z_0}} {{z_0}}}}}, {\mathbbm{y}}\!\mapsto\!{\tobedone {{\clos {\Omega} {\![]}}}}] &&\stackrel{(\ref{tr:8})}{\to}\\[-0.6ex]
26:&& {\tapp {{\tapp {c} {({\tlam {{z_0}} {{z_0}}})}}} {{\cut {{\tlam {{z_0}} {{z_0}}}}_\ctrian}}} &| [] &&\stackrel{(\ref{tr:10})}{\to}\\[-0.6ex]
27:&& {\cut {{\tapp {{\tapp {c} {({\tlam {{z_0}} {{z_0}}})}}} {{\tlam {{z_0}} {{z_0}}}}}}_\ctrian} &| []\phantom{\stackrel{(1)}{\to}} && \nrightarrow
\end{align*}
}%
\end{table}

We apply a convention that locations introduced by
transitions~(\ref{tr:2}), concerning memoization of normal forms of abstractions,
are taken from the sequence
$\mathbbm a, \mathbbm b, \mathbbm d, \mathbbm e$ and those introduced
by transitions~(\ref{tr:6}) and~(\ref{tr:7}), concerning memoization of values
of arguments, from  $\mathbbm x, \mathbbm y, \mathbbm z, \mathbbm w$.
To shorten the text repeating in Table~\ref{fig:execution_ex_ref}, we
define an auxiliary term
$A~:=~{\tlam {y} {{\tlam {z} {{\tapp {I} {z}}}}}}$, an
environment
$e_\mathit{yz} := {\![y\!\mapsto\!{\mathbbm{y}},
  z\!\mapsto\!{\mathbbm{z}}]}$, a
closure
$x^{\mathbbm{x}} := {\clos {x} {\![x\!\mapsto\!{\mathbbm{x}}]}}$, and
a store\break
$\sigma_1 := [{\mathbbm{x}}\!\mapsto\!{\done
  {{{{\mathbbm{d}}}\!:=\!{{\clos {{\tlam {z} {{\tapp {I} {z}}}}}
          {\![y\!\mapsto\!{\mathbbm{y}}]}}}}}},
{\mathbbm{y}}\!\mapsto\!{\tobedone {{\clos {\Omega} {\![]}}}}]$.

We also employ here \emph{refocusing notation}.  It is a general
technique to shorten the presentation of an abstract machine and to
make it more readable.  It means that instead of writing the
configuration explicitly (the following is the configuration after step
14):
$$\langle {{\clos {{\tapp {I} {z}}} {\!e_\mathit{yz}}}}, {{\cons {{\tlam {{z_0}} {\hole}}} {{\cons {{{{\mathbbm{d}}}\!:=\!{\hole}}} {{\cons {{\tapp {c} {\hole}}} {{\cons {{\tapp {\hole} {x^{\mathbbm{x}}}}} {[]}}}}}}}}}, {\,\sigma_1 \!\ast\! [{\mathbbm{d}}\!\mapsto\!{\tobedone {\bot}}, {\mathbbm{z}}\!\mapsto\!{\done {{z_0}}}]}\rangle_\etrian$$
we recompose the term and the stack into a processed term and use the angle brackets to mark the place of decomposition:
${\tapp {{\tapp {c} {\left({{{\mathbbm{d}}}\!:=\!{{\tlam {{z_0}} {{\cut {{\clos {{\tapp {I} {z}}} {\!e_\mathit{yz}}}}_\etrian}}}}}\right)}}} {x^{\mathbbm{x}}}} | \,\sigma_1 \!\ast\! [{\mathbbm{d}}\!\mapsto\!{\tobedone {\bot}}, {\mathbbm{z}}\!\mapsto\!{\done {{z_0}}}]$.
The store remains on the right-hand side after a vertical bar.

We consistently omit locations in the store that are not reachable
from the recomposed term.  They can be thought of as garbage collected
by reference counting.  This is why the location~$\mathbbm a$
disappears just after its allocation and the store is empty in the
end.  Otherwise, the location~$\mathbbm a$ would contain the
$\tobedone \bot$ till the end of the execution.

In a higher-level perspective, the machine performs three
$\beta$-reductions: in step 3 the term
${\tapp {({\tlam {x} {{\tapp {{\tapp {c} {x}}} {x}}}})} (\tapp A
  \Omega)}$ is reduced to
${{\tapp {{\tapp {c} {(\tapp A \Omega)}}} {(\tapp A \Omega)}}}$; in
step 11 the latter term, which is
${{\tapp {{\tapp {c} {(\tapp {(\tlam {y} {{\tlam {z} {{\tapp {I}
                    {z}}}}})} \Omega)}}} {(\tapp A \Omega)}}}$, is
reduced to
${{\tapp {{\tapp {c} {({\tlam {z} {{\tapp {I} {z}}}})}}}
    {(\tapp A \Omega)}}}$; in step 17 the last term is reduced to
${{\tapp {{\tapp {c} {({{\tlam {z} { {z}}}})}}} {(\tapp A
      \Omega)}}}$. The argument $\tapp A \Omega$ is normalized in steps
8--21, memoized in step 22 and reused in step 26.

In a lower-level perspective, the machine first moves its focus to
$({\tlam {x} {{\tapp {{\tapp {c} {x}}} {x}}}})$ (step~1), recognizes
it as an abstraction and gives it a location~$\mathbbm a$ (step~2).
It immediately applies the abstraction to closure
${{\clos {{\tapp {A} {\Omega}}} {\![]}}}$, so the closure moves to the
store with a fresh location~$\mathbbm x$ and location~$\mathbbm a$ disappears
because it is no longer reachable (step~3).  Then the environment
$[x\!\mapsto\!{\mathbbm{x}}]$ percolates through the applications to
the variables (steps~4--5).  The variable $c$ turns out to be free, so
it is considered as a value (step 6) and the focus moves to the first
remaining $x$ (step~7).  It turns out to be substituted by an
unevaluated closure whose value will be memoized at
location~$\mathbbm{x}$ (step~8).  Then the focus moves to $A$ (step~9),
which is recognized as an abstraction (step~10) and immediately
applied to the closure ${{\clos {\Omega} {\![]}}}$, so $A$'s body is evaluated
in the extended environment (step~11).  The body is also an
abstraction (step~12).  The abstraction is memoized at
location~$\mathbbm{x}$ (step~13).  It has no arguments, so it is
normalized: the parameter $z$ is renamed to a fresh name $z_0$, and
the normal form of the abstraction will be memoized at location~$\mathbbm{d}$ (step~14).
Note that $\mathbbm{d}\!:=\!{\hole}$ is outside the focus, so it is a frame on the stack now. 
The
focus moves to $I$ and the environment $e_\mathit{yz}$ percolates
(step~15).  The identity is an abstraction (step~16), and it is
applied to closure ${{\clos z {e_\mathit{yz}}}}$, so the closure moves to the
store with a fresh location~$\mathbbm w$ (step~17).  The
variable $x$ constituting the identity's body looks up the
location~$\mathbbm{w}$ containing the unevaluated argument (step~18).
By subsequent lookup, it is evaluated to $z_0$ (step~19) and memoized
at $\mathbbm{w}$ (but it is unreachable and disappears) (step~20).
The abstraction $\tlam {z_0} {z_0}$ is reconstructed (step~21) and
memoized at $\mathbbm{d}$ as a normal form of
${{\tlam {z} {{\tapp {I} {z}}}}}$ (step~22).  The machine reconstructs
${{\tapp {c} {{\tlam {{z_0}} {{z_0}}}}}}$ (step~23) and moves focus to
the second $x$ (step~24).  The variable $x$ looks up the
location~$\mathbbm{x}$ with the value of
${{\clos {{\tapp {A} {\Omega}}} {\![]}}}$ (step~25) that looks up the
location~$\mathbbm{d}$ for its normal form by
transition~(\ref{tr:8})\,(step~26).  Finally, the result
${\tapp {{\tapp {c} {({\tlam {{z_0}} {{z_0}}})}}} {{\tlam {{z_0}}
      {{z_0}}}}}$ is reconstructed (step~27).

To summarize, the application $\tapp A \Omega$ is evaluated only once, its
value is normalized only once, and the nontermination of $\Omega$ is
contained.  Occurrences of $x$ (in $\tapp {\tapp c x} x$, $I$, and
$\Omega$) do not collide, the bound variable $z$ is renamed, and the
free variable $c$ stays untouched.

\subsection{Empirical execution lengths}
\label{sec:empirical}

Abstract cost models, including abstract machines, give us an
opportunity to measure the program complexity independently of the
efficiency of a physical machine.
In Table~\ref{tab:numbers}, for six term families, we present the
number of steps required to reach the normal form for:
\begin{enumerate}
\item the normal-order strategy,
\item KN, implemented as in Table 4 of \cite{Cregut:HOSC07},
\item RKNL, and
\item \rknle, which is an exemplary modification of RKNL that applies the
transition~(\ref{tr:7}) unconditionally, leaving the transition~(\ref{tr:8})
unused.
\end{enumerate}

The closed forms are verified for natural numbers from 1 to 9 because
some of the sequences have a different value for $n = 0$.
The numbers of machine steps were measured by running the machines KN, RKNL,
and \rknle and counting their steps.
The number of the normal-order strategy steps was measured by
counting $\beta$-steps in KN executions.

Table \ref{tab:numbers} presents formulae for the sequences with explicit
constants. The formulae are not approximate, but they give the exact numbers.
Thanks to that, these exact results should be easily reproducible.
However, comparison of constants between, for example, KN and RKNL is not
crucially informative because they depend on the granularity of the transition
rules. The formulae are given to show the asymptotic behaviour of the
sequences.

The first two families are the families from Table 9 of
\cite{Cregut:HOSC07} up to renaming.  We name a~few more terms to
define these families ($\KI$ is the normal form of term $\tapp K I$,
and $c_n$ is the $n^\text{th}$ Church numeral).
\begin{align*}
&K := \tlam x {\tlam y x}, \;\;
\KI := \tlam x {\tlam y y}, \;\;
\pair := \tlam x {\tlam y {\tlam f {\tapp {\tapp {\;f} x} y}}}, \;\;
\dub := \tlam x {\tlam f {\tapp {\tapp {\;f} x} x}}\\
&\pred := \tlam n {\tlam f {\tlam x {
\tapp
    {\;\left(\tapp
        {\tapp {n}
               {\left(\tlam e {\tapp {\tapp {\;\pair} {(\tapp e \KI)}} {(\tapp f {(\tapp e \KI)})}} \right)}}
        {(\tapp {\tapp \pair x} x)}\right)}
K}}}\\
&c_n := \tlam f {\tlam x {\underbrace{f(f(...f(x)...))}_{f~\text{applied}~n~\text{ times}}}}
\hspace{1cm}
d_n := \begin{cases}I &: n = 0 \\ \tlam v {\tapp {\tapp {(\tlam x {\tlam k {\tapp k {\tlam f {\tapp {\tapp {\;f} x} x}}}})} v} {d_{n-1}}} &: n > 0\end{cases}
\end{align*}

\begin{table}[h]
\caption{Empirical execution lengths for $1 \leq n \leq 9$}
\label{tab:numbers}
\def\arraystretch{1.5}
\setlength\tabcolsep{4mm}
\begin{center}
{\small{
\begin{tabular}{ |c|c c c c| }
 \hline
 term family & $\stackrel{\no}{\to}$ & KN & RKNL & \rknle \\
 \hline
 $\tapp {\tapp {c_n} {c_2}} I$ &$3\cdot 2^n -1$& $15 \cdot 2^n - 6$ &  \hspace{-1em}$10 \cdot 2^n + 5n + 5$ & $10 \cdot 2^n + 5n + 5$ \\
 $\tapp {\pred} {c_n}$ &$6n+8$& $26n + 25$ & $30n + 41$ & $30n + 41$\\
 $\tlam x {\tapp {\tapp {c_n} \omega} x}$ &$2^n + 1$& $12 \cdot 2^n - 3$ & $9n + 15$& $9n + 15$ \\
 $\tapp {\tapp {c_n} {\dub}} I$ &$2^n + 1$& $23 \cdot 2^n - 14$ & $18n + 15$ & $16 \cdot 2^n + 5n - 1$ \\
 \hspace{-1em}$\tapp {\tapp {c_n} {\dub}} {(\tlam x {\tapp I x})}$\hspace{-1em} &$2\cdot 2^n + 1$& $26 \cdot 2^n - 14$ & $18n + 20$ & $21 \cdot 2^n + 5n - 1$\\
 $\tapp {d_n} I$ & $3n + 1$ & $22 \cdot 2^n + 7n - 15$ \hspace{-1em} & $28n + 10$ & $16 \cdot 2^n + 15n - 6$ \\
 \hline
\end{tabular}}}
\end{center}
\end{table}

The family $\tapp {\tapp {c_n} {c_2}} I$ employs exponentiation of
Church numerals, so it composes $2^n$ identity functions. It takes
an exponential number of steps in normal order. All three machines
reproduce this exponential behaviour, but RKNLs seem to reduce the
coefficient of the exponential element.

The family $\tapp {\pred} {c_n}$ computes the predecessor of a Church
numeral. The overhead of memoization of RKNL with respect to the number
of steps of KN is visible. Nevertheless, RKNL is asymptotically not worse than
KN.  It is intuitive because of RKNL construction, and possibly it could be
formally proven, but here we limit ourselves to empirical evidence supporting
this hypothesis.

The family $\tlam x {\tapp {\tapp {c_n} \omega} x}$ doubles with $\omega$
the number of $x$ variables, $n$ times. While KN constructs the output term
step by step, RKNL achieves normal form in a number of steps logarithmic
w.r.t. the normal-order strategy, so it is {\em implosive} in terms of
\citet{DBLP:conf/lics/AccattoliCC21}.

Similarly, the family $\tapp {\tapp {c_n} {\dub}} I$ constructs pairs with the
same element, $n$ times. Here, thanks to the memoization, and in consequence
sharing, of strong normal forms of functions, RKNL remains implosive.
In \rknle, the transition~(\ref{tr:8}) is incapacitated, so the machine is not
able to reuse the normal forms of the functions, and recomputes them at
exponential cost, w.r.t. RKNL.

The family $\tapp {\tapp {c_n} {\dub}} {(\tlam x {\tapp I x})}$ differs from
the previous one by $\eta$-expanded identity copied by normal order deep down
in the intermediate term. It adds $2^n$ redexes to reduce, and ends up with
doubled coefficient of the exponential element in the number of
$\stackrel{\no}{\to}$ steps. However, it costs RKNL only 5 steps more.

Finally, the family $\tapp {d_n} I$ is a simplified call-by-value
continuation-passing style (CPS) translation of $\tapp {\tapp {c_n} {\dub}} I$.
They compute the same normal forms. The enforced call-by-value order avoids
call-by-name duplication of arguments and results in linear number of steps in
the normal-order reduction. RKNL, being call-by-need, should be asymptotically
not worse than call-by-value reduction and maintains the linear number of
steps, which is reasonable. \rknle\, cannot reuse the normal forms of the
functions, and turns out to be unreasonable w.r.t. the normal-order strategy.

\section{Properties of the machine}
\label{sec:proofs}
In this section we discuss the correctness of the derived machine and
its complexity.

\subsection{Decoding}
\label{sec:decoding}

We say that a configuration $k$ is \emph{reachable} if there exists a
sequence of machine transitions starting in an initial configuration
and ending in $k$.  We define a decoding of reachable configurations
(denoted $\deck{\,\cdot\,}$) to terms in order to refer micro-step
operational semantics of the machine to small-step operational
semantics of normal-order reduction.  Formally, our goal is to show
that whenever $k \to k'$ for a reachable configuration $k$, then
$\deck{k} \stackrel{\no}{\twoheadrightarrow}=_\alpha \deck{k'}$. We
prove it in Lemma~\ref{lem:step_soundness}.

\begin{table}[ht]%
\caption{Decoding of RKNL}
\label{fig:decoding}
\begin{align*}
\decc{\clos {\tapp {t_1} {t_2}} e, \Sigma\!} &:= \tapp {\decc{\clos {t_1} e, \Sigma\!}} {\decc{\clos {t_2} e, \Sigma\!}}\\
\decc{\clos {\tlam x t} e, \Sigma\!} &:= \tlam {\overline x} {\decc{\clos t {\update e x {\overline x}}, \Sigma\!}}\\
\decc{\clos x e, \Sigma\!} &:= \begin{cases}
\overline x &: e(x) = \overline x \mbox{ defined by the rule above}\\
\decc{\clos {t_2} {e_2}, \Sigma\!} &: \text{$ e(x) $ initialized as $\tobedone{\clos {t_2} {e_2}}$ with (\ref{tr:6}) in } \Sigma\\
\fresh x &: \text{$ e(x) $ initialized as $\hspace{8.5mm} \done{\fresh x}$ with  (\ref{tr:7}) in }\Sigma \\
x &: e(x) \mbox{ not defined and not initialized } 
\end{cases}
\end{align*}
\begin{equation*}
\begin{aligned}[c]
\decv {t, \Sigma\!} &:= t\\
\decv {\abs {\tlam x t} e \ell, \Sigma\!} &:= \decc {{\clos {\tlam x t} e}, \Sigma\!}\\
~\\
\deck{{\econf \sigma t e s}, \Sigma} &:= \plug {\decs{s, \Sigma\!}} {\decc{\clos t e, \Sigma\!}} \\
\deck{{\cconf \sigma   v s}, \Sigma} &:= \plug {\decs{s, \Sigma\!}} {\decv{v, \Sigma\!}}\\
\deck{k} &:=\deck{k,\Sigma_k}
\end{aligned}
\hspace{10mm}
\begin{aligned}[c]
\decs{\nil, \Sigma\!} &:= \hole\\
\decs{\cons {\frapp {\clos t e}} s, \Sigma\!} &:= \plug {\decs{s, \Sigma\!}} {\frapp {\decc{\clos t e, \Sigma\!}}}\\
\decs{\cons {\flapp t} s, \Sigma\!} &:= \plug {\decs{s, \Sigma\!}} {\flapp t}\\
\decs{\cons {\flam x} s, \Sigma\!} &:= \plug {\decs{s, \Sigma\!}} {\flam x}\\
\decs{\cons {\fcache \ell} s, \Sigma\!} &:= \decs{s, \Sigma\!}
\end{aligned}
\end{equation*}
\end{table}

Our decoding is presented in Table \ref{fig:decoding}. It uses
decodings $\decc{\cdot}$ of closures, $\decv{\cdot}$ of values, and
$\decs{\cdot}$ of stacks. For the most part, the decoding is simple
and standard (the currently processed term or closure is decoded and
plugged in the context obtained from the decoding of the stack),
however the case of variables requires an explanation.  As mentioned
above, we are interested in the correspondence between our machine and
the normal-order strategy. Since the latter recomputes function
arguments that occur more than once in the function body, there are
moments in time when some of the occurrences are already computed and
the others are not. To keep track of this, in the decoding of a
variable we are not interested in the current state of the store, but
instead we prefer the state of the store from the moment of the
initialization of the variable.  Therefore each of the decoding
functions carries an additional argument $\Sigma_k$, namely the full
sequence of configurations from the initial one to $k$.  This
additional argument is also implicitly present in decodings $\deck{k}$
of configurations; we omit it only to shorten the notation.

Another subtle point is the treatment of free and bound variables. We
treat them differently in order to be able to handle open terms. Free
variables are never added to any environment and thus they remain
untouched in the decoding. On the other hand, bound variables are
overlined. We assume here that the set of overlined variables forms a
disjoint copy of the original set of variables. This way the closure
$\big((\lambda x. \lambda y. x y) y ,[]\big)$ decodes to
$(\lambda \overline{x}. \lambda \overline{y}.\overline{x}
\overline{y}) y$ and its reduct
$\big((\lambda y. x y),[x \mapsto y]\big)$ decodes to
$\lambda \overline{y}. y \overline{y}$ thus avoiding the capture of variable~$y$.
  
In summary, the four cases in the decoding of
$\decc{\clos x e, \Sigma}$ correspond to the following four situations: (i)
$x$ is bound and not initialized in $\Sigma$; (ii) $x$ is bound and
initialized with (\ref{tr:6}); (iii) $x$ is bound and initialized with
(\ref{tr:7}); (iv) $x$ is free.

The decodings of configurations from the elaborate example execution
from Table~\ref{fig:execution_ex_ref} are shown in
Table~\ref{fig:decodings_execution_ex_ref} (decoding 1).
\subsubsection{Alternative decodings}
\label{alternative_decodings}
This subsection is extraneous to the main content of the paper.
We provide it to the interested reader who wants to explore
design choices of the decoding and bridges to other research developments.

We consider two alternative decodings that are independent of the execution history. We enumerate all of the decodings we consider:
\begin{enumerate}
\item The decoding described above, dependent on the execution history.
\item Same as above, but decode variables to values that are currently in the store.
\item Same as above, but decode variables to terms
reconstructed below corresponding $\fcache \ell$ frames.  Then the
transition~(\ref{tr:6}) would be decoded to parallel $\beta$-reduction
in multicontexts as in \cite{DBLP:conf/lics/AccattoliCC21} and the
transition~(\ref{tr:5}) would not change the decoded term.
\end{enumerate}

\begin{table}[htb]
\caption{Decodings of the elaborate example execution}
\label{fig:decodings_execution_ex_ref}
{
\begin{align*}
\text{step}\phantom{:} && \text{decoding 1} && \text{decoding 2} && \text{decoding 3} \\ 
0\text{--}2:&& {\tapp {({\tlam {\overline {x}} {{\tapp {{\tapp {c} {\overline {x}}}} {\overline {x}}}}})} {{(\tapp {\overline{A}} {\overline{\Omega}})}}} && {\tapp {({\tlam {\overline {x}} {{\tapp {{\tapp {c} {\overline {x}}}} {\overline {x}}}}})} {{(\tapp {\overline{A}} {\overline{\Omega}})}}} && {\tapp {({\tlam {\overline {x}} {{\tapp {{\tapp {c} {\overline {x}}}} {\overline {x}}}}})} {{(\tapp {\overline{A}} {\overline{\Omega}})}}} \\
3\text{--}10:&& {\tapp {{\tapp {c} {{(\tapp {\overline{A}} {\overline{\Omega}})}}}} {{(\tapp {\overline{A}} {\overline{\Omega}})}}} && {\tapp {{\tapp {c} {{(\tapp {\overline{A}} {\overline{\Omega}})}}}} {{(\tapp {\overline{A}} {\overline{\Omega}})}}} && {\tapp {{\tapp {c} {{(\tapp {\overline{A}} {\overline{\Omega}})}}}} {{(\tapp {\overline{A}} {\overline{\Omega}})}}} \\
11\text{--}12:&& {\tapp {{\tapp {c} {({\tlam {\overline {z}} {{\tapp {\overline{I}} {\overline {z}}}}})}}} {{(\tapp {\overline{A}} {\overline{\Omega}})}}} && {\tapp {{\tapp {c} {({\tlam {\overline {z}} {{\tapp {\overline{I}} {\overline {z}}}}})}}} {{(\tapp {\overline{A}} {\overline{\Omega}})}}} && {\tapp {{\tapp {c} {({\tlam {\overline {z}} {{\tapp {\overline{I}} {\overline {z}}}}})}}} {{\tlam {\overline {z}} {{\tapp {\overline{I}} {\overline {z}}}}}}} \\
13:&& {\tapp {{\tapp {c} {({\tlam {\overline {z}} {{\tapp {\overline{I}} {\overline {z}}}}})}}} {{(\tapp {\overline{A}} {\overline{\Omega}})}}} && {\tapp {{\tapp {c} {({\tlam {\overline {z}} {{\tapp {\overline{I}} {\overline {z}}}}})}}} {{\tlam {\overline {z}} {{\tapp {\overline{I}} {\overline {z}}}}}}} && {\tapp {{\tapp {c} {({\tlam {\overline {z}} {{\tapp {\overline{I}} {\overline {z}}}}})}}} {{\tlam {\overline {z}} {{\tapp {\overline{I}} {\overline {z}}}}}}} \\
14\text{--}16:&& {\tapp {{\tapp {c} {{(\tlam {{z}_0} {{\tapp {\overline{I}} {{z}_0}}}}}})} {{(\tapp {\overline{A}} {\overline{\Omega}})}}} && {\tapp {{\tapp {c} {{(\tlam {{z}_0} {{\tapp {\overline{I}} {{z}_0}}}}}})} {{\tlam {\overline {z}} {{\tapp {\overline{I}} {\overline {z}}}}}}} && {\tapp {{\tapp {c} {{(\tlam {{z}_0} {{\tapp {\overline{I}} {{z}_0}}}}}})} {{\tlam {{z}_0} {{\tapp {\overline{I}} {{z}_0}}}}}} \\
17\text{--}24:&& {\tapp {{\tapp c {({\tlam {{z}_0} {{z}_0}})}}} {{(\tapp {\overline{A}} {\overline{\Omega}})}}} && {\tapp {{\tapp c {({\tlam {{z}_0} {{z}_0}})}}} {{\tlam {\overline {z}} {{\tapp {\overline{I}} {\overline {z}}}}}}} && {\tapp {{\tapp c {({\tlam {{z}_0} {{z}_0}})}}} {{\tlam {{z}_0} {{{{z}_0}}}}}} \\
25:&& {\tapp {{\tapp c {({\tlam {{z}_0} {{z}_0}})}}} {{\tlam {\overline {z}} {{\tapp {\overline{I}} {\overline {z}}}}}}} && {\tapp {{\tapp c {({\tlam {{z}_0} {{z}_0}})}}} {{\tlam {\overline {z}} {{\tapp {\overline{I}} {\overline {z}}}}}}} && {\tapp {{\tapp c {({\tlam {{z}_0} {{z}_0}})}}} {{\tlam {{z}_0} {{{{z}_0}}}}}} \\
26\text{--}27:&& {\tapp {{\tapp c {({\tlam {{z}_0} {{z}_0}})}}} {{\tlam {{z}_0} {{z}_0}}}} && {\tapp {{\tapp c {({\tlam {{z}_0} {{z}_0}})}}} {{\tlam {{z}_0} {{z}_0}}}} && {\tapp {{\tapp c {({\tlam {{z}_0} {{z}_0}})}}} {{\tlam {{z}_0} {{z}_0}}}} \\
\end{align*}
}%
\end{table}

Table \ref{fig:decodings_execution_ex_ref} presents the three decodings for
the configurations from Table \ref{fig:execution_ex_ref}.
For brevity, we naturally define the following overlined terms as the same terms
with overlined variables:
$\overline{I} := \tlam {\overline{x}} {\overline{x}}$,
$\overline A := {\tlam {\overline{y}} {{\tlam {\overline{z}} {{\tapp {\overline{I}} {\overline{z}}}}}}}$, and
$\overline\Omega := \overline\omega\overline\omega$ with
$\overline\omega := \tlam {\overline{x}} {\tapp {\overline{x}} {\overline{x}}}$.

The three decodings agree on the initial configuration and first ten steps.
In particular, they decode step 3, performing the transition (\ref{tr:6}),
to the same normal-order $\beta$-reduction:
${\tapp {({\tlam {\overline {x}} {{\tapp {{\tapp {c} {\overline {x}}}} {\overline {x}}}}})} {{(\tapp {\overline{A}} {\overline{\Omega}})}}}
\stackrel{\no}{\to}
{\tapp {{\tapp {c} {{(\tapp {\overline{A}} {\overline{\Omega}})}}}} {{(\tapp {\overline{A}} {\overline{\Omega}})}}}$.

The first difference between decodings 1 and 2 appears after step 13,
performing transition~(\ref{tr:5}), which under the location
${\mathbbm{x}}$ puts the value
${\done {{{{\mathbbm{d}}}\!:=\!{{\clos {{\tlam {z} {{\tapp {I} {z}}}}}
          {\![y\!\mapsto\!{\mathbbm{y}}]}}}}}}$.  Here, in the
history-based decoding 1, closure $x^{\mathbbm{x}}$ is decoded to
$\tapp {\overline{A}} {\overline{\Omega}}$ (which is in line with
initialization in step 3), while in the store-based decoding 2 it
decodes to
$\tlam {\overline {z}} {{\tapp {\overline{I}} {\overline {z}}}}$.  On
one hand, decoding~2 has a less involved definition than the decoding
1 because it is history-free, i.e., it does not depend on the full
sequence of configurations of a reachable configuration (denoted
$\Sigma$).  On the other hand, the behaviour of the decoded term
degrades.  In particular, step 13 decodes to a $\beta$-reduction which
is not a normal-order reduction:
${\tapp {{\tapp {c} {({\tlam {\overline {z}} {{\tapp {\overline{I}}
              {\overline {z}}}}})}}} {{(\tapp {\overline{A}}
      {\overline{\Omega}})}}}
\not{\phantom{.}}\hspace{-3mm}\stackrel{\no}{\to} {\tapp {{\tapp {c}
      {({\tlam {\overline {z}} {{\tapp {\overline{I}} {\overline
                {z}}}}})}}} {{\tlam {\overline {z}} {{\tapp
          {\overline{I}} {\overline {z}}}}}}}$.  Moreover,
transition~(\ref{tr:5}), which is silent (i.e., it changes nothing) in
decoding 1, is decoded into a simultaneous replacement of terms in
potentially many places in decoding~2.  Again, the simplicity of the
definition of decoding is traded for the simple interpretation of
machine transitions to the degree that decoding 2 can be considered
ugly.

An interesting compromise is decoding 3, which does not depend on the
reduction history as decoding 2, but its formal definition is more involved because
it requires chasing ${\fcache \ell}$ frames on the stack.  The
information on how to decode the term is more dispersed.  However,
thanks to that, transition (\ref{tr:6}) can be read as a parallel
$\beta$-reduction.  In step 11, it happens in the multicontext
${\tapp {\tapp c \hole} \hole}$, and, in step 17, in the multicontext
${\tapp {{\tapp {c} {{(\tlam {{z}_0} \hole}}})} {{\tlam {{z}_0}
      \hole}}}$.  Similarly, transition~(\ref{tr:7}) can be read
as a parallel $\alpha$-conversion, as in step 14.  Transitions
(\ref{tr:4}) and (\ref{tr:8}), which are bypass transitions in
decoding~1 (Lemma~\ref{lem:bypass}), are silent in decoding~3.  In
fact, (\ref{tr:6}) and (\ref{tr:7}) are the only transitions that
change the decoded term.  Moreover, the parallel reduction can reduce
an arbitrarily big number of redexes at once.

A downside of decoding 3 is that the
machine execution is not seen as a simulation of a normal-order
reduction, and the relation of the machine to the normal-order
strategy is less obvious.  This moves some difficulty of the proof of
correctness of the machine from the decoding to the relation between
the parallel reduction and normal order, which is more theoretical
ground. Such an approach was taken in the study of call by value
\cite{DBLP:conf/lics/AccattoliCC21}.
Moreover, the departure from the normal-order strategy does not help
to relate the machine to the call-by-need calculi that perform single
$\beta$-reduction steps because of the parallel nature of the reduction
in decoding 3.

\paragraph*{Alternative approach to soundness.}
One might consider targeting a strong call-by-need strategy---rather
than normal order---as the reference semantics that could potentially
simplify the soundness argument by separating sharing concerns from
reduction concerns. We believe this would shift the proof burden
rather than reduce it. Strong CbNd does not admit a clean reduction
semantics in the pure $\lambda$-calculus (cf. discussion in
Section~\ref{sec:cbnd_strategy}): any formalization requires syntactic
extensions such as explicit substitutions or let-constructs,
introducing their own technical overhead. Normal order, by contrast,
has a standard and well-understood reduction semantics that serves as
a simple reference point. In particular, the discussion around
decoding 3 above hints at some of the complexity such an approach
would require.

\subsection{Ghost abstract machine for RKNL}
One of the goals in the design of the normalizer of
Listing~\ref{fig:cbnd-normalizer} was to keep the derived abstract
machine simple. Indeed, the machine has only 11 transitions, that is 5
transitions more than the machine for weak call by need (recalled in
Section~\ref{sec:conservative}) and as many as 13 transitions fewer
than the known machine for strong call by need
in~\cite{BiernackaC19}. One of the design choices to achieve this
goal was not to include the grammar of normal forms in the syntactic
category of values.  However, to reason about correctness, we need to
know that only terms in normal form may appear as values.

Therefore,
for theoretical purposes,
we introduce a ghost abstract machine RKNLi---a more refined
version of RKNL with explicitly expressed invariants---presented in Table~\ref{fig:shape_invariant}.
It can be seen as a proof technique tailored to
reason about invariants of abstract machines (in this case, RKNL, in particular in Lemmas~\ref{lem:stack}~and~\ref{lem:increase}).

The ghost machine bisimulates RKNL so that every reachable
configuration of RKNL is a projection of a reachable configuration of
the ghost machine. The name \textit{ghost} is analogous to ghost variables
and ghost code from systems such as Why3. Such variables can be added to
the executable code just to state and prove its invariants. However, their
values are not computed during real execution. Thus, the executable code can be
seen as a projection of the code with explicitly stated invariants.

\begin{table}[htbp]
\caption{The RKNLi ghost abstract machine, RKNL with explicit shape invariant}
\label{fig:shape_invariant}
\begin{alignat*}{3}
\mathit{Identifiers}   \ni &\;& x &\\
\mathit{Terms}         \ni && t   &::= x \alt \tapp{t_1}{t_2} \alt \tlam x t\\
\mathit{Normal\;Terms} \ni && n   &::= \tlam x n \alt \lceil a \rceil\\
\mathit{Neutral\;Terms} \ni && a   &::= x \alt \tapp{a}{n}\\
\mathit{Locations} \ni && \ell\\
\mathit{Envs}   \ni && e   & \;\subtype \fun {\mathit{Identifiers}} {\mathit{Locations}} \\
\mathit{Closures} \ni && c   &::= \clos t e\\
\mathit{Locations'} \ni && \ell'\\
\mathit{Values} \ni && v   &::= a \alt \abs {\tlam x t} e {\ell'}\\
\mathit{Storable\;Values} \ni && \_ & ::= \tobedone c \alt \done v\\
\mathit{Stores} \ni && \;\sigma & \;\subtype \fun {\mathit{Locations}} {\mathit{Storable\;Values}}\\
\mathit{Optional\;Normal\;Terms} \ni && \_ & ::= \tobedone \bot \alt \done n\\
\mathit{Stores'} \ni && \;{\sigma'} & \;\subtype \fun {\mathit{Locations'}} {\mathit{Optional\;Normal\;Terms}}\\
\mathit{Potentially\;Applicative\;Stacks}    \ni && \pi &::= \cons {\fcache \ell} \pi
	\alt \cons {\frapp c} \pi
	\alt \cons {\lceil \hole \rceil} \varrho \\
\mathit{Non\text{-}applicative\;Stacks} \ni && \varrho    &::=  \cons {\fcache {\ell'}} \varrho
	\alt \nil
	\alt \cons {\flam x} \varrho
	\alt \cons {\flapp a} \pi \\
\mathit{Confs} \ni && k   &::=
  \langle c, \pi, {\sigma, \sigma'} \rangle_\etrian
  \alt \cconf {\sigma, \sigma'} v \pi
  \alt \cpconf {\sigma, \sigma'} n \varrho\\
\mathit{Transitions:}\\[-3em]
\end{alignat*}%

\begin{align}
   t &\mapsto \econf {\nil, \nil} t \nil {\cons {\lceil \hole \rceil} \nil}\nonumber\\
\econf {\sigma, \sigma'} {\tapp{t_1}{t_2}} e {\hspace{16mm}\pi} &\to \econf         {\sigma, \sigma'}     {t_1} e {\cons{\frclos{t_2} e} {\pi}} \tag{$1$}\label{tri:1}\\
\econf {\sigma, \sigma'}       {\tlam x t} e {\hspace{16mm}\pi} &\to \cconf {\sigma, \alloc {\sigma'} {\ell'} {\tobedone \bot}} {\abs {\tlam x t} e {\ell'}} \pi \tag{$2$}\label{tri:2}\\
  \econf {\sigma, \sigma'}              x    e {\hspace{16mm}\pi} &\to \econf {\sigma, \sigma'} t {e_2} {\cons {\fcache \ell} \pi}\nonumber\\
  & \hspace{19mm}\mbox{where } \ell={\lookup e x}, \;{\lookup \sigma {\ell}} = \tobedone {\clos t {e_2}}  \tag{$3$}\label{tri:3}\\
  \econf {\sigma, \sigma'}              x    e {\hspace{16mm}\pi} &\to \cconf {\sigma, \sigma'} v \pi \nonumber\\
     & \hspace{19mm}\mbox{where } {\lookup \sigma {\lookup e x}} = \done v  \vee (v = x \notin e) \tag{$4$}\label{tri:4}\\
\cconf {\sigma, \sigma'} v {\hspace{6mm}\cons {\fcache \ell} \pi} &\to \cconf {\update \sigma \ell {\done v}, \sigma'} v \pi \tag{$5$}\label{tri:5}\\
\cpconf {\sigma, \sigma'} n {\hspace{4mm} \cons {\fcache {\ell'}} \varrho} &\to \cpconf {\sigma, \update {\sigma'} {\ell'} {\done n}} n \varrho \tag{$5'$} \label{tri:5'}\\
\cconf {\sigma, \sigma'} {\abs {\tlam x t} e {\ell'}} {\cons {\frclos{t_2} {e_2}} \pi} &\to \econf {\alloc \sigma {\ell_2} {\tobedone {\clos{t_2} {e_2}}}, \sigma'} t {\update e x {\ell_2}} \pi \tag{$6$}\label{tri:6}\\
\cconf {\sigma, \sigma'} {\abs {\tlam x t} e {\ell'}} {\hspace{7.7mm}{\cons {\lceil \hole \rceil} \varrho}} &\to \langle {\clos t {\update e x {\ell_2}}}, {\cons {\lceil \hole \rceil} {\cons {\flam {\fresh x}} {\cons {\fcache {\ell'}} \varrho}}}, \nonumber\\
& \hspace{10.6mm} {\alloc \sigma {\ell_2} {\done {\fresh x}}, \sigma'} \rangle_\etrian
\hspace{2.5mm}\mbox{where } \lookup {\sigma'} {\ell'} = \tobedone \bot \tag{$7$}\label{tri:7} \\
\cconf {\sigma, \sigma'} {\abs {\tlam x t} e {\ell'}} {\hspace{7.7mm}{\cons {\lceil \hole \rceil} \varrho}} &\to \cpconf {\sigma, \sigma'} n \varrho \hspace{14.5mm}\mbox{where } \lookup {\sigma'} {\ell'} = \done n \tag{$8$}\label{tri:8}\\ 
\cconf {\sigma, \sigma'} a                         {\cons {\frclos{t_2} {e_2}} \pi} &\to \econf {\sigma, \sigma'} {t_2} {e_2} {\cons {\lceil \hole \rceil} {\cons {\flapp a} \pi}}\tag{$9$}\label{tri:9}\\
\cconf {\sigma, \sigma'} {a} { \hspace{7.7mm} \cons {\lceil \hole \rceil} \varrho} &\to \cpconf {\sigma, \sigma'} {\lceil a \rceil} \varrho \tag{$9a$} \label{tri:9a}\\
\cpconf {\sigma, \sigma'} {n} {\hspace{7.3mm}\cons {\flapp {a}} \pi} &\to \cconf {\sigma, \sigma'} {\tapp {a} {n}} \pi \tag{$10$}\label{tri:10}\\ 
\cpconf {\sigma, \sigma'} n {\hspace{5.3mm}\cons {\flam x} \varrho} &\to \cpconf {\sigma, \sigma'} {\tlam x n} \varrho \tag{$11$}\label{tri:11}\\
\cpconf {\sigma, \sigma'} n {\hspace{14mm}\nil} &\mapsto n \nonumber
\end{align}
\end{table}

The main difference between RKNL and RKNLi is the grammar of normal
terms in RKNLi. Note that each neutral term is in normal form. For
technical reasons we want to distinguish between neutral terms and
normal forms, so we add an explicit coercion: a neutral term $a$
tagged with $\lceil\cdot\rceil$ becomes a normal term
$\lceil a\rceil$. Introduction of normal terms triggers further
refinements in stacks and configurations. Now we have two variants of
stacks: potentially applicative stacks may have an argument on top, while
non-applicative stacks certainly do not have arguments on top. This in
turn introduces two variants of $\ctrian$-configurations: one with
potentially applicative stacks and one with non-applicative stacks. Yet another
difference is the introduction of a separate store: RKNLi has two
stores, one for evaluated arguments and one for
normalized lambda abstractions.

Changes in the grammars lead to changes in transitions. RKNLi has two
versions of transition~(\ref{tr:5}): one for normal terms and one for
values (which includes neutral terms). It also has an additional
transition~(\ref{tri:9a}) that implements the coercion from neutral to
normal terms: after finishing the evaluation of an argument of a
neutral term initialized by transition~(\ref{tri:9}) (as well as after
evaluating the body of a lambda abstraction initialized by
transition~(\ref{tri:7}) and after completing the whole evaluation),
the result being a neutral term must be coerced to a normal term.
From the point of view of the RKNL machine, this step is silent.
Lemma~\ref{lem:shape} states formally the correspondence between the
two machines.

\begin{remark}\label{ghost_derivation}
  Good news about ghost machines as a proof technique is that ghost machines
  neither need to be designed by hand, but can be derived from a higher-order
  normalizer. However, such a derivation is easier to track if presented in a
  statically typed language where data structures are defined by algebraic
  data types.

  In the normalizer, being the starting point of the derivation, such as one in
  Listing~\ref{fig:no-normalizer}, the output term can be represented by
  the type expressing the grammar of normal terms. This type contains an
  explicit coercion from neutral terms to normal terms. In
  Table~\ref{fig:shape_invariant}, the coercion is $\lceil a \rceil$ from the
  grammar of $n$. The continuations obtained by CPS translation appropriately
  handle normal and neutral terms, which is reflected in their types. After
  defunctionalization, they form the algebraic data types of the stacks,
  expressing the shape invariant. The coercion is imprinted on the grammar of
  stacks as the $\lceil \hole \rceil$ frame. Similarly to previous steps,
  functions that appropriately handle different types of stacks give rise to
  more types of machine configurations.

  The ghost machine was not obtained from the {\tt semt} tool because it works
  on an untyped language and thus the invariants would not be reflected in the
  types. However, an analogous tool for a statically typed language would be
  useful to derive ghost machines faster.

  An additional example of a ghost machine, analogous to RKNLi, was prepared
  by~\citet{Drab23} for the KN machine \cite{Cregut:HOSC07}, while the code
  accompanying~\cite{BiernackaBCD20} contains a~derivation of a ghost machine
  (called invariant derivation) for the KNV machine introduced there.
\end{remark}

\subsection{Soundness}
In this section we prove the soundness of the RKNL machine with
respect to the normal-order strategy. Note that we do not rely here on
the correctness of the transformations used in the {\tt semt} tool as
the tool is not formally verified. Instead, we directly show that the
machine faithfully realizes the strategy.

The result is proved by a~sequence of lemmas showing that each
transition of the machine is sound, where soundness of a transition
$k\to k'$ is understood as reducibility (up to $\alpha$-renaming) of
the decoding of $k$ to the decoding of $k'$.  Note that since
$\beta$-reduction commutes with $\alpha$-conversion
(cf. Example~\ref{ex:alpha}), a sequence of sound transitions is sound.

\begin{lemma}[load correctness]\label{lem:load}
$\deck{\econf \nil t \nil \nil} =_\alpha t$. 
\end{lemma}
\begin{proof}
  As the closure $ \clos t {\nil}$ is
  decoded recursively, the environment is extended only by overlined
  bound variables of the source term. While bound source variables
  are decoded to overlined versions of themselves and bound by the
  appropriate abstractions, the free variables and the structure of
  the term remain untouched, resulting in an $\alpha$-renamed version of
  the loaded term.
\end{proof}

\begin{lemma}[overhead transitions]
  \label{lem:overhead}
If $k \stackrel{(\iota)}{\to} k'$ then $\deck k = \deck {k'}$ for $\iota \notin \{\ref{tr:4}, \ref{tr:6}, \ref{tr:7}, \ref{tr:8}\}$.
\end{lemma}
\begin{proof}
  In transition~(\ref{tr:1}), the structural change is that the right part of
  the application is peeled off and put on the stack, while the environment
  percolates to both parts. This directly corresponds to the decoding.
  Store modifications of transitions~(\ref{tr:2}) and (\ref{tr:5}) do not
  affect the decoding. Transition~(\ref{tr:3}) replaces a variable by
  a closure with the same decoding. Transitions~(\ref{tr:9}), (\ref{tr:10}) and
  (\ref{tr:11}) only recompose the processed term.
\end{proof}

\begin{lemma}[alpha transition]\label{lem:alpha}
  If $k$ is a reachable configuration of RKNL and
  $k \stackrel{(\ref{tr:7})}{\to} k'$, then $ \deck k =_\alpha \deck {k'}$.
\end{lemma}

\begin{proof}
  The configuration $k$ is of the form $\cconf \sigma {\abs {\tlam x t} e {\ell}} s$
  and $k'$ is of the form $\econf {\alloc \sigma {\ell_2} {\done {\fresh x}}} t {\update e x {\ell_2}} {\cons {\flam {\fresh x}}{\cons {\fcache \ell} s}}$, 
  where $\lookup \sigma \ell = \tobedone \bot$.
  Let $C$ be the context obtained from the decoding of the stack $s$ in
  $k$. Then $\plug C {\flam {\fresh x}}$ is the decoding of the stack
  in $k'$. The decoding of the value $\abs {\tlam x t} e \ell$ in $k$
  is
  $\tlam {\overline x} {\decc{\clos t {\update e x {\overline x}},
      \Sigma_k\!}}$, where all occurrences of $x$ are decoded to
  $\overline x$.  On the other hand, the decoding of the closure in
  $k'$ is $\decc{\clos t {\update e x {\ell_2}}, \Sigma_{k'}}$. Since
  $\Sigma_{k'}$ differs from $\Sigma_{k}$ only by the last
  configuration initializing $\ell_2$ with $\fresh x$, this gives the
  same term with $\overline x$ replaced by $\fresh x$.  Hence
  $\plug C {\decv{\abs {\tlam x t} e \ell, \Sigma_k}}$ is
  $\alpha$-equivalent to
  $\plug C {\plug {\lambda {\fresh x.}} {\decc{\clos t {\update e x {\ell_2}},
        \Sigma_{k'}}}}$, which means that $ \deck k =_\alpha \deck {k'}$.
\end{proof}

In the following, by a projection of a stack ($\pi$ or $\varrho$) of RKNLi we
mean the stack obtained by removing all occurrences of the frame
$\lceil \hole \rceil$, and erasing all tags $\lceil \cdot \rceil$ from
normal terms. The pair of stores in RKNLi corresponds to the single
store in RKNL, defined roughly as their union. This notion of
projection lifts naturally to terms, configurations and all the
grammars occurring in the definition of RKNLi in
Table~\ref{fig:shape_invariant}.

The following lemma is proved by a simple induction and case analysis
of all transitions. In particular it implies that there are two types
of stacks $s$ occurring in reachable configurations of RKNL:
projections of stacks generated by the grammars of $\pi$ and $\varrho$ from the
definition of RKNLi.

\begin{lemma}[shape invariant]
\label{lem:shape}
Every reachable configuration of RKNL is a projection of a
reachable configuration of RKNLi.
\end{lemma}

\begin{lemma}[stack shape invariant]
  \label{lem:stack}
  For any reachable configuration $k$ with stack $s$, the stack
  decodes to a normal-order context: $\decs{s, \Sigma_k\!} \in \NO$.
\end{lemma}

\begin{proof}
  By Lemma~\ref{lem:shape} the stack $s$ is a projection of an
  $\pi$-stack or a $\varrho$-stack of RKNLi. Frames of the form $\fcache \ell$ are ignored
  in the decoding of Table~\ref{fig:decoding}. The grammars of $\pi$
  and $\varrho$ stacks, with frames of this form removed, correspond directly to the
  inside-out grammar of normal-order contexts from
  Section~\ref{sec:reduction}, with $\pi$ corresponding to $N$ and $\varrho$ corresponding to $\underline{N}$.
\end{proof}

\begin{samepage}

\begin{lemma}[beta transition]\label{lem:beta}
  If $k$ is a reachable configuration of RKNL and
  $k \stackrel{(\ref{tr:6})}{\to} k'$, then
  $ \deck k \stackrel{\no}{\to} \deck {k'}$.
\end{lemma}

\begin{proof}
   Configuration $k$ is of the form
  $\cconf \sigma {\abs {\tlam x t} e {\ell}} {\cons {\frclos{t'} {e'}} s}$ and 
  the configuration $k'$ is of the form $ \econf {\alloc \sigma {\ell_2} {\tobedone {\clos{t'} {e'}}}} t {\update e x {\ell_2}} s $.
  Let $C$ be the context obtained by decoding  stack $s$ in
  $k$. Then $k$ decodes into the
  application 
  $\tapp {\big( \tlam {\overline x} {\decc{\clos t {\update e x
        {\overline x}}, \Sigma_{k}}}\big)} {\decc{\clos{t'} {e'},
    \Sigma_k\!}}$ in context $C$, and $k'$ decodes into
  $\decc{\clos {t} {\update e x {\ell_2}}, \Sigma_{k'}}$ in the same context.

  Since $\ell_2$ is initialized to $\tobedone{\clos{t'} {e'}}$ in
  $\Sigma_{k'}$, and
  $\decc{\clos {t'}{e'},\Sigma_{k'}}=\decc{\clos {t'}{e'},\Sigma_k}$,
  we have
  $$\decc{\clos {t} {\update e x {\ell_2}}, \Sigma_{k'}}=\decc{\clos
    {\subst x {\decc{\clos {t'}{e'},\Sigma_k}} {t}}
    {e_{- x}}, \Sigma_k}.$$
  Here
  $e_{-x}$ denotes the environment $e$ without the definition for $x$.    
   
  \noindent This implies that 
  $\deck {k'}= \plug C {\decc{\clos {\subst x {\decc{\clos
            {t'}{e'},\Sigma_k}} {t}} {e_{- x}}, \Sigma_k}}$. We have observed that $\deck{k}=\plug C {\tapp {\big( \tlam {\overline x} {\decc{\clos t {\update e x {\overline x}}, \Sigma_{k}}}\big)}
    {\decc{\clos{t'} {e'}, \Sigma_k\!}} } $, hence
  $ \deck k {\to} \deck {k'}$ is an instance of
  $\beta$-reduction. By
  Lemma~\ref{lem:stack} we know that $C$ is a normal-order context, hence we have
  $ \deck k \stackrel{\no}{\to} \deck {k'}$.
\end{proof}
\end{samepage}

By a \emph{bypass transition} we mean a~transition that reads
a~memoized value from the store, that is transition~(\ref{tr:4}) or
(\ref{tr:8}). We say that a bypass transition ($\iota$) \emph{accesses}
a~location $\ell$ if
\begin{itemize}
\item $\iota = \ref{tr:8}$ and $\ell$ is the location associated with the
  lambda closure, or
\item $\iota = \ref{tr:4}$ and $\ell = e(x)$ is the location assigned to the
  variable $x$ in the environment (not mentioned explicitly in the
  transition).
\end{itemize}

The next lemma is proved by induction on locations, using the
following order. We define $\ell\sqsubseteq\ell'$ if $\ell'$ was
populated with a value after $\ell$. More formally, for a location
$\ell$, let $\Sigma_\ell$ be the sequence of configurations starting
from the initial configuration and ending with a configuration of the
form $\cconf {\update \sigma \ell {\done v}} v s$ reached by
transition~(\ref{tr:5}) that updates $\ell$. Then
$\ell\sqsubseteq\ell'$ if $\Sigma_\ell$ is a prefix of
$\Sigma_{\ell'}$. In the case when $\Sigma_\ell$ is not defined (which
is possible, e.g., the value is not needed, as for $\mathbbm{y}$ in
Table~\ref{fig:execution_ex_ref}), $\ell$ remains incomparable with
other locations.  For eventually populated locations, $\sqsubseteq$ is a
well-founded linear order.

\begin{lemma}[bypass transitions]
  \label{lem:bypass}
  For all locations $\ell$, if $k$ is a reachable configuration of RKNL and
  $k \stackrel{(\iota)}{\to} k'$ by a bypass transition accessing $\ell$, then
  $ \deck k \stackrel{\no}{\twoheadrightarrow}=_\alpha \deck {k'}$.
\end{lemma}

\begin{proof}
%
  The location $\ell$ must have been created fresh, by one of the transitions:~(\ref{tr:2}),~(\ref{tr:6}) or~(\ref{tr:7}). Since none of
  transitions in Table~\ref{fig:machine} overwrites a storable value
  of the form $\done v$, such a value is uniquely determined by the
  location.
  We now examine each case of how the location $\ell$ was created.
\begin{itemize}  
\item{Case 1:} $\ell$ was created by transition~(\ref{tr:7}) (where it is
  called $\ell_2$).  Then $\ell$ was filled immediately with a value (an
  $\alpha$-renamed variable). Therefore, in this case $(\iota)$ is (\ref{tr:4}),
  and any later access to $\ell$ simply retrieves this memoized
  value. Hence, the decoding does not change, and we have
  $\deck{k} = \deck{k'}$.
\item{Case 2:} $\ell$ was created by transition~(\ref{tr:2}) or (\ref{tr:6}). In this case, $\ell$ did not initially contain a value, and was later updated via transition~(\ref{tr:5}). Let 
$k_{\ell}$ be the configuration where the frame $\ell := \square$ was pushed (i.e., in transition~(\ref{tr:3}) or (\ref{tr:7})), and $k_{\ell}^v$ be the configuration where the value $v$ was stored into $\ell$.

The sequence $k_{\ell}\rred{}{} {k_{\ell}^v}$ contains no transitions
(\ref{tr:4}) nor (\ref{tr:8}) accessing $\ell$, and thus by
Lemmas~\ref{lem:overhead}--\ref{lem:beta} and by induction hypothesis
we obtain
$ \deck {k_\ell} \stackrel{\no}{\twoheadrightarrow}=_\alpha \deck {k_{\ell}^v}$.

We now have two cases to consider.
\begin{itemize}
\item{Case $\iota = 4$:} In this case $k=\econf \sigma x e s \quad\text{and}\quad
  k'=\cconf \sigma v s$.
  The frame $\fcache \ell$ was pushed with
  transition~(\ref{tr:3}), so for
  $\ell=e(x)$, and for some stack $s'$ and
  stores $\sigma', \sigma''$ we have
  \[k_\ell= \econf {\sigma'} x {e} {{s'}} \mbox{ and }
    k_{\ell}^v= \cconf {\sigma''} v {\cons {\fcache \ell} {s'}}.\]

Then
$$\deck{k_\ell} = \plug {\decs{s',\Sigma_{k_\ell}}}{\decc{\clos {x}{e},\Sigma_{k_\ell}}}\;\text{and}\;
  \deck{k_{\ell}^v}=\plug {\decs{s',\Sigma_{k_{\ell}^v}}}{\decv{v,\Sigma_{k_{\ell}^v}}}.$$ Moreover,
  $\decs{s',\Sigma_{k_{\ell}^v}}=\decs{s',\Sigma_{k_\ell}}$, which forms a
  normal-order context by Lemma~\ref{lem:stack}. Since we have
  $ \deck {k_\ell} \stackrel{\no}{\twoheadrightarrow}=_\alpha \deck {k_{\ell}^v}$,
  it follows that
  $\decc{\clos {x}{e},\Sigma_{k_\ell}} \stackrel{\no}{\twoheadrightarrow}=_\alpha \decv{v,\Sigma_{k_{\ell}^v}}$.
  
  Finally, we have $\decs{s,\Sigma_{k}}= \decs{s,\Sigma_{k'}}$ and
  $\decv{v,\Sigma_{k_{\ell}^v}}=\decv{v,\Sigma_{k'}}$, and hence 
  \[
  \deck k = \plug {\decs{s,\Sigma_{k}}}{\decc{\clos {x}{e},\Sigma_{k}}}
  \stackrel{\no}{\twoheadrightarrow}=_\alpha
  \plug {\decs{s,\Sigma_{k'}}}{\decv{v,\Sigma_{k'}}} = \deck {k'}.
  \]

\item{Case $\iota = 8$:}  In this case
    $k= \cconf \sigma {\abs {\tlam x t} e {\ell}} s$ and
    $k'=\cconf \sigma v s$. The frame  $\fcache \ell$ was pushed with
    transition~(\ref{tr:7}), so for some stack $s'$ and stores
    $\sigma',\sigma''$ we have
     \[k_\ell= \cconf {\sigma'}  {\abs {\tlam x t} {e}{\ell}} {s'}
       \mbox{ and }
       k_{\ell}^v= \cconf {\sigma''} v {\cons {\fcache \ell} {s'}}.\] 
       
   Then
    $\deck{k_\ell} = \plug {\decs{s',\Sigma_{k_\ell}}}{\decc{\clos {\tlam x t}{e},\Sigma_{k_\ell}}}$ and
    $\deck{k_{\ell}^v}=\plug {\decs{s',\Sigma_{k_{\ell}^v}}}{\decv{v,\Sigma_{k_{\ell}^v}}}$. Moreover,
    $\decs{s',\Sigma_{k_{\ell}^v}}=\decs{s',\Sigma_{k_\ell}}$, which forms a 
    normal-order context by Lemma~\ref{lem:stack}. Since we have
    $ \deck {k_\ell} \stackrel{\no}{\twoheadrightarrow}=_\alpha \deck {k_{\ell}^v}$,
    it follows that
    $\decc{\clos {\tlam x t}{e},\Sigma_{k_\ell}} \stackrel{\no}{\twoheadrightarrow}=_\alpha \decv{v,\Sigma_{k_{\ell}^v}}$.

    Finally, we have $\decs{s,\Sigma_{k}}= \decs{s,\Sigma_{k'}}$ and
    $\decv{v,\Sigma_{k_{\ell}^v}}=\decv{v,\Sigma_{k'}}$. The context $\decs{s,\Sigma_{k}}$ is a
    normal-order context by Lemma~\ref{lem:stack}, hence we obtain that
    \[ \deck {k}= \plug {\decs{s,\Sigma_{k}}}{\decc{\clos {\tlam x t}{e},\Sigma_{k}}}
      \stackrel{\no}{\twoheadrightarrow}=_\alpha\plug {\decs{s,\Sigma_{k'}}}{\decv{v,\Sigma_{k'}}}=
      \deck {k'}.
    \]
 \end{itemize}   
\end{itemize}
\end{proof}

Observe that it is possible that a bypass transition
$k \stackrel{(\ref{tr:4})}{\to} k'$ does not access any location: this
is the case when $x$ is a free variable that does not occur in the
environment $e$ (the second option in the transition). In this case
$\decc{\clos {x}{e}}=\decv x$, so $\deck{k} = \deck{k'}$.

The following lemma is an immediate consequence of Lemmas~\ref{lem:overhead}--\ref{lem:bypass} and the observation above.

\begin{lemma}[step soundness]\label{lem:step_soundness}
  If $k$ is a reachable configuration of RKNL and $k \to k'$ then
  $\deck{k} \stackrel{\no}{\twoheadrightarrow}=_\alpha \deck{k'}$.
\end{lemma}

\begin{lemma}[unload correctness]\label{lem:unload}
For any terminal reachable configuration $\cconf \sigma t \nil \!\nrightarrow$,
the configuration decodes to a term $t$ in normal form: $\deck{\cconf \sigma t \nil} = t \!\nrightarrow_\beta$.
\end{lemma}
\begin{proof}
  By Lemma~\ref{lem:shape} (shape invariant), a~terminal configuration
  $\cconf \sigma t \nil$ of RKNL is a projection of a terminal
  configuration $\cpconf {\sigma, \sigma'} n \nil$ of
  RKNLi. Therefore term $t$ is a projection of the normal term $n$,
  and thus it is in normal form. Finally,
  $\deck{\cconf \sigma t \nil} =\plug { \decs{\nil}} {\decv{t}} = t$.
\end{proof}

\noindent
Lemmas~\ref{lem:load}--\ref{lem:unload} lead to the main result of
this section.

\begin{theorem}[soundness]\label{thm:sound}
If the machine starts from a (possibly open) term $t_0$ and outputs a term~$t$
(i.e., $\econf \nil {t_0} \nil \nil \twoheadrightarrow \cconf \sigma t \nil \!\nrightarrow$),
then $t_0$ reduces to a normal form~$t$ in the normal order reduction in zero or more steps
(i.e., $t_0 \stackrel{\no}{\twoheadrightarrow}=_\alpha t \!\nrightarrow_\beta$).
\end{theorem}

\subsection{Complexity}

We now analyse the complexity of the RKNL machine using the method of
potentials to perform amortized cost analysis.  We define the
potential function $\Phik$ on configurations to estimate the number of
steps the machine in a given configuration can make without performing
$\beta$-transition~(\ref{tr:6}). This function is defined together
with auxiliary potential functions for terms $\Phit$, values $\Phiv$,
stacks $\Phis$ and stores $\Phih$ in Table~\ref{fig:potential}. An
empirical observation of the behaviour of the potential function is
presented in the dedicated subsection \ref{sec:empirical_potential}.

The potential of a configuration is defined as the sum of potentials
of its components. The potential of a term is defined as the sum of
potentials of its subterms, with an additional constant for each term
constructor. The potential of a value is a constant, while the
potential of a stack is the sum of potentials of its frames, with an
additional constant for each frame constructor. The potential of a
store is the sum of potentials associated with two sets of terms: the
first one is the set of closures passed as function arguments and not
yet evaluated, and the second one is the set of abstraction bodies
that have not yet been normalized. When these terms will be evaluated,
their potential will be accounted for in the potential of the term or
stack.

Each constant occurring in the potential of term and stack
constructors accounts for its later use and can be traced
backwards. For example, the application constructor
$\tapp {t_1} {t_2}$ is replaced by a right frame $\frapp {\clos t e}$
in the transition $(\ref{tr:1})$, the right frame turns into a left
frame $\flapp t$ in the transition $(\ref{tr:9})$, and finally the
left frame is discharged by the transition $(\ref{tr:10})$. Thus, in
the potential function for the stack, we give 1 credit for the left
frame to pay for $(\ref{tr:10})$, which results in the definition
$\Phis(\cons {\flapp t} s) := 1 + \Phis(s)$. Further, we give 2
credits for the right frame: one to pay for $(\ref{tr:9})$, and one to
be forwarded to the left frame. Analogously, we assign 3 credits to
the application constructor in the potential function for terms.


\begin{table}[h]%
\caption{The potential function for RKNL}
\label{fig:potential}
\begin{align*}%
\Phik(\econf \sigma t e s) &:= \Phit(t) + \Phis(s) + \Phih(k) \\
\Phik(\cconf \sigma   v s) &:= \Phiv(v) + \Phis(s) + \Phih(k) 
\end{align*}%
\begin{equation*}
\begin{aligned}[c]
\Phit(\tapp {t_1} {t_2}) &:= 3 + \Phit(t_1) + \Phit(t_2)\\
\Phit(\tlam x t) &:= 4 + \Phit(t)\\
\Phit(x) &:= 2\\
\Phiv(t) &:= 0\\
\Phiv({\abs {\tlam x t} e \ell}) &:= 1
\end{aligned}
\hspace{10mm}
\begin{aligned}[c]
\Phis(\nil) &:= 0\\
\Phis(\cons {\frapp {\clos t e}} s) &:= 2 + \Phis(s) + \Phit(t)\\
\Phis(\cons {\flapp t} s) &:= 1 + \Phis(s)\\
\Phis(\cons {\flam x} s) &:= 1 + \Phis(s)\\
\Phis(\cons {\fcache \ell} s) &:= 1 + \Phis(s)
\end{aligned}
\end{equation*}
\begin{align*}%
  \Phih(k) &:= \sum_{\ell \in k \wedge \sigma(\ell) = \tobedone {\clos t e} \wedge \ell := \hole \notin s} \!\!\Phit(t)
             + \sum_{\ell := {\clos {\tlam x t} e} \in k \wedge \sigma(\ell) = \tobedone \bot \wedge \ell := \hole \notin s} \!\!(2 + \Phit(t))
\end{align*}%
\end{table}

\noindent
The number of consecutive
non-$\beta$-transitions from a given configuration is bounded by the two
following lemmas.

\begin{lemma}[decrease]\label{lem:decrease}
If $k \stackrel{(\neq\ref{tr:6})}{\to} k' $ then $ \Phik(k) > \Phik(k')$.
\end{lemma}
\begin{proof} 
  The proof proceeds by case analysis on transition rules
  (\ref{tr:1}--\ref{tr:11}) with the omission of transition
  (\ref{tr:6}):
  \begin{description}
  \item[(\ref{tr:1})]
    $\begin{array}[t]{rcl}\Phik(\econf \sigma {\tapp{t_1}{t_2}} e s)&=&
      (3 + \Phit(t_1) + \Phit(t_2)) + \Phis(s) + \Phih(k) \\&>&
           \Phit(t_1) + (2 + \Phis(s) + \Phit(t_2)) + \Phih(k)\\&=&
\Phik(\econf         \sigma     {t_1} e {\cons{\frclos{t_2} e} {s}})
     \end{array}$ 
   \item[(\ref{tr:2})]
     $\begin{array}[t]{rcl}\Phik(\econf \sigma       {\tlam x t} e s)
        &=&
            (4 + \Phit(t)) + \Phis(s) + \Phih(k)\\
        &>&
            1 + \Phis(s) + \left(\Phih(k) + \left(2 + \Phit(t)\right)\right)\\
        &=&
            \Phik( \cconf {\alloc \sigma \ell {\tobedone \bot}} {\abs {\tlam x t} e \ell} s )                                  \end{array}$
   \item[(\ref{tr:3})]
     $\begin{array}[t]{rcl}\Phik( \econf \sigma              x    e s)
        &=&
          2 + \Phis(s) + \Phih(k)  \\  
        &=&
          2 + \Phis(s) + (\Phih(k') + \Phit(t)) \\  
        &>&
           \Phit(t) + (1 + \Phis(s)) + \Phih(k')  \\
        &=&
         \Phik( \econf \sigma t {e_2} {\cons { \fcache {\ell}} s} )   
      \end{array}$\\
      (here location $\ell$ contributes $\Phit(t)$ to the store
      potential of configuration $k$, but it does not contribute to
      the store potential of $k'$ as the frame $\fcache \ell$ appears in
      the new stack; the evaluation of $t$ starts, so its potential is
      removed from the potential of the store and instead it appears as the
      potential of the evaluated term)
   \item[(\ref{tr:4})]
     $\begin{array}[t]{rcl}\Phik(\econf \sigma              x    e s )
        &=&
          2 + \Phis(s) + \Phih(k)\\   
        &>&
           \Phiv(v) + \Phis(s) + \Phih(k)\\ 
        &=&
         \Phik( \cconf \sigma v s )   
  \end{array}$
   \item[(\ref{tr:5})]
     $\begin{array}[t]{rcl}\Phik(\cconf \sigma v {\cons {\fcache \ell} s} )
        &=&
          \Phiv(v) + (1 + \Phis(s)) + \Phih(k)\\  
        &>&
          \Phiv(v) + \Phis(s) + \Phih(k)\\  
        &=&
         \Phik( \cconf {\update \sigma \ell {\done v}} v s )   
      \end{array}$\\
      (here the frame $\fcache \ell$ is removed from the stack so it does not contribute to the store potential anymore)
   \item[(\ref{tr:7})]
     $\begin{array}[t]{rcl}\Phik(\cconf \sigma {\abs {\tlam x t} e {\ell}} s)\!\!\!
        &=&
         1 + \Phis(s) + \Phih(k) \\   
         &=&
         1 + \Phis(s) + (\Phih(k') + 2 + \Phit(t))\\   
        &>&
          \Phit(t) + 1 + 1 + \Phis(s) + \Phih(k')\\  
        &=&
         \Phik(\econf {\alloc \sigma {\ell_2} {\done {\fresh x}}} t {\update e x {\ell_2}} {\cons {\flam {\fresh x}} {\cons {\fcache \ell} s}} )   
      \end{array}$\\
      (here location $\ell$ contributes $2 + \Phit(t)$ to the store
      potential of configuration $k$, but it does not contribute to
      the store potential of $k'$ as the frame $\fcache \ell$ appears in
      the new stack; the evaluation of $t$ starts, so its potential is
      removed from the potential of the store and instead it appears as the
      potential of the evaluated term)
   \item[(\ref{tr:8})]
     $\begin{array}[t]{rcl}\Phik(\cconf \sigma {\abs {\tlam x t} e \ell} s)
        &=&
           1 + \Phis(s) + \Phih(k)\\ 
        &>&
           0 + \Phis(s) + \Phih(k)\\ 
        &=&
         \Phik(  \cconf \sigma v s)   
  \end{array}$
   \item[(\ref{tr:9})]
     $\begin{array}[t]{rcl}\Phik(\cconf \sigma t                         {\cons {\frclos{t_2} {e_2}} s} )
        &=&
            0 + (2 + \Phis(s) + \Phit(t_2)) + \Phih(k)\\
        &>&
           \Phit(t_2) + (1 + \Phis(s)) + \Phih(k)\\ 
        &=&
         \Phik(\econf \sigma {t_2} {e_2} {\cons {\flapp t} s} )   
  \end{array}$
   \item[(\ref{tr:10})]
     $\begin{array}[t]{rcl}\Phik( \cconf \sigma {t_2} {\cons {\flapp {t_1}} s})
        &=&
          0 + (1 + \Phis(s)) + \Phih(k)\\  
        &>&
          0 + \Phis(s) + \Phih(k)\\  
        &=&
         \Phik(\cconf \sigma {\tapp {t_1} {t_2}} s )   
  \end{array}$
   \item[(\ref{tr:11})]
     $\begin{array}[t]{rcl}\Phik(\cconf \sigma t {\cons {\flam x} s} )
        &=&
          0 + (1 + \Phis(s)) + \Phih(k)\\  
        &>&
           0 + \Phis(s) + \Phih(k)\\ 
        &=&
         \Phik(\cconf \sigma {\tlam x t} s )   
  \end{array}$
\end{description}
\end{proof}

\begin{lemma}[increase]\label{lem:increase}
  If $k$ is a~configuration reachable from an input term~$t_0$, then
  $k \stackrel{(\ref{tr:6})}{\to} k'$ implies $\Phik(k) + \Phit(t_0) >
  \Phik(k')$.
\end{lemma}
\begin{proof} Table~\ref{fig:shape_invariant} shows that abstractions
  are only constructed when rebuilding normal forms, and they cannot
  be applied.  Therefore, any applied abstraction has to be a~subterm
  of the source term, hence 
  $\Phit(t_0) > \Phit(t)$, where $t$ is the body of an applied
  abstraction in transition~(\ref{tr:6}). Looking closely at the
  respective stores in both configurations, we can see that
  $\Phih(k') \le \Phih(k) + \Phit(t_2)$, as $t_2$ is passed as
  a~function argument in $k$ and not yet evaluated in $k'$. Thus, we
  have
  \[\begin{array}[t]{rcl}
  \Phik(k) + \Phit(t_0) &=& \Phik(\cconf \sigma {\abs {\tlam x t} e {\ell}} {\cons {\frclos{t_2} {e_2}} s}) + \Phit(t_0)\\
                        &=& 1 + (2 + \Phis(s) + \Phit(t_2)) + \Phih(k) + \Phit(t_0) \\
                        &>& \Phit(t) + \Phis(s) + (\Phih(k) + \Phit(t_2))\\
&\ge&\Phik(\econf {\alloc \sigma {\ell_2} {\tobedone {\clos{t_2} {e_2}}}} t {\update e x {\ell_2}} s)\\
 &=& \Phik(k')
\end{array}
\]\end{proof}

We can summarize the two lemmas above by the following theorem stating that the machine
simulates the normal-order strategy in a bilinear number of steps,
\ie, linear in the number of $\beta$-reductions and linear in the size $|t_0|$
of the input term $t_0$---because the potential of the input term is bounded
by its size times a constant factor.

\begin{theorem}\label{thm:bilinearity}
Let $\rho$ be a sequence of consecutive machine transitions starting
from term $t_0$ to configuration $k'$. Let
$|\rho|$ be the number of steps in $\rho$
and $|\rho|_\beta$ be the number of normal-order $\beta$-reductions from $t_0$ to $\deck{k'}$.
Then ${|\rho| \leq (|\rho|_\beta + 1) \cdot \Phit (t_0)}$.
\end{theorem}

\subsubsection{Implementation in the RAM model}
\label{sec:ram}

Theorem~\ref{thm:bilinearity} shows that our machine simulates the
normal-order strategy in a bilinear number of steps. In this section
we show that each step in this simulation can be implemented on a
random-access machine in time logarithmic in the size of the input
term, so the overall simulation cost is quasibilinear. More
precisely, in the RAM model the running time of our simulation is
linear in the number of $\beta$-steps and quasilinear (i.e., of order
$n \log n$) in the size of the initial term. This improves the result of 
Accattoli et al.~\cite{Accattoli16,DBLP:conf/lics/AccattoliCC21}, who 
proved that normal-order strategy can be simulated with Useful Milner
Abstract Machine with quadratic overhead.

The key observation is the following lemma, which is a simple
consequence of the distinction between the syntactic category of
\textit{Terms} (ranged over by $t$) and the categories of
\textit{Normal Terms} and \textit{Neutral Terms} (ranged over by $n$
and $a$, respectively) in the ghost machine. The only possibility to
introduce a new term to a configuration of the machine is to load
it. The loaded term may be then decomposed into subterms by
transitions~(\ref{tri:1}), (\ref{tri:6}) and~(\ref{tri:7}); these
subterms may be copied into the stack (transition~(\ref{tri:1})) or to
the store (transition~(\ref{tri:6})) and then read from there by
transitions (\ref{tri:3}), (\ref{tri:6}) and~(\ref{tri:9}), but terms
in this category are never constructed. New terms always belong to the
categories of \textit{Normal Terms} or \textit{Neutral Terms}, and are
constructed only by transitions (\ref{tri:7}) (new variables),
(\ref{tri:9a}), (\ref{tri:10}) and~(\ref{tri:11}). The second part of
the lemma follows from the first one and Lemma~\ref{lem:shape}.

\begin{lemma}[subterm property]\label{lem:subterm} Every term in the
  syntactic category of \textit{Terms} that occurs in any reachable
  configuration of RKNLi is a subterm of the source term. Every
  term~$t$ that occurs in a closure $\clos t e$ that occurs in any
  reachable configuration of RKNL is a subterm of the source term.
\end{lemma}

The only non-local operations performed by the RKNL machine are on
environments and on the store.  The store is never duplicated, so it
can be represented in the standard way by heap-allocated locations; we
can therefore assume that store operations run in constant time.

By Lemma~\ref{lem:subterm} environments are finite dictionaries
associated only with subterms of the source term. Consequently, their
size is bounded by the number of variables occurring in the source
term. A balanced tree implementation yields logarithmic time
complexity for operations on environments, with respect to the number
of entries.  Alternatively, a trie-based implementation yields time
proportional to the length of the longest identifier in the source
term. This gives the following corollary.

\begin{corollary}
  RKNL is a reasonable abstract machine. It can be simulated on a
  random-access machine running in time
  $O((|\rho|_\beta + 1) \cdot |t_0| \log |t_0|)$, where $t_0$ is the
  initial term and $|\rho|_\beta$ is the number of $\beta$-reductions in the
  normal-order normalization of $t_0$.
\end{corollary}

In practice, the number of distinct source variables is small; if it
can be treated as a constant, then the $\log |t_0|$ factor collapses
to a constant.  Further improving practical performance would require
empirical profiling of the RKNL abstract machine or the corresponding
normalizer in Listing~\ref{fig:cbnd-normalizer}.

\subsubsection{Examples of the potential function development}\label{sec:empirical_potential}

We can inspect the behaviour of the potential function by plotting it over an execution.
Figure~\ref{fig:plot_of_example} displays the configuration and store potentials
for the running example of Table~\ref{fig:execution_ex_ref}.
For reference, the potentials of some relevant subterms are:
$\Phit(I) = 6,\;
 \Phit(\omega) = 11,\;
 \Phit(\Omega) = 25,\;
 \Phit(\tapp A \Omega) = 47,\;
 \Phit({\tlam {x} {{\tapp {{\tapp {c} {x}}} {x}}}}) = 16$.
The store potential counts
only those locations that actually occur in the configurations listed in
Table~\ref{fig:execution_ex_ref}. 
Observe that in this example the overall potential never increases:
each abstraction is either applied or normalized, but never both,
so its own potential budget pays for the cost of its corresponding
$\beta$-reduction.

\begin{figure}[h]
\begin{center}
\includegraphics[scale=.4]{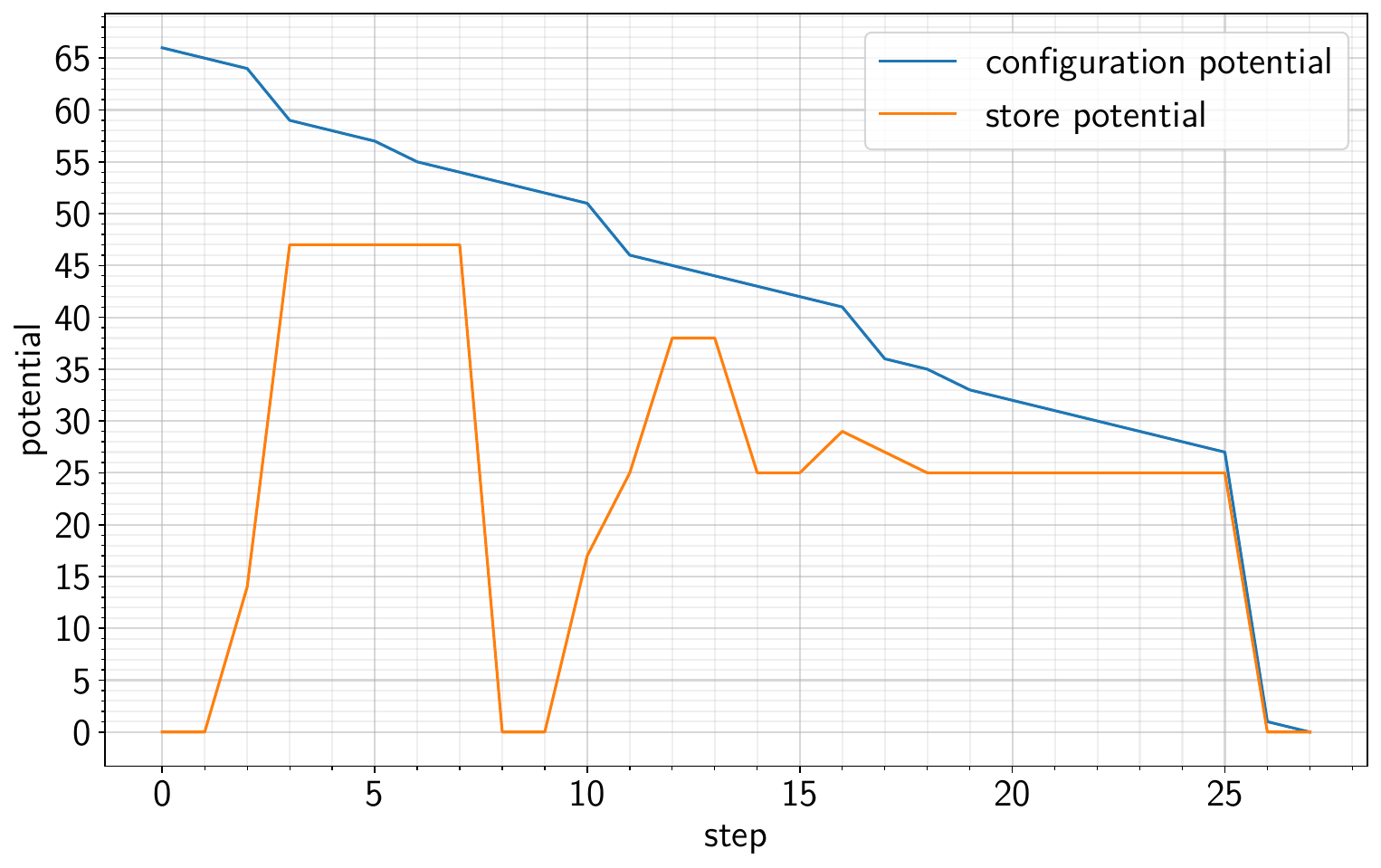}
\caption{Plot of potentials of configurations and stores
  for the execution of term
  ${\tapp {({\tlam {x} {{\tapp {{\tapp {c} {x}}} {x}}}})} {\big({\tapp {({\tlam {y} {{\tlam {z} {{\tapp {I} {z}}}}}})} {\Omega}}\big)}}$ as presented in
  Table~\ref{fig:execution_ex_ref}. }
\label{fig:plot_of_example}
\end{center}
\end{figure}

\begin{figure}[h]
\begin{center}
\includegraphics[scale=.4]{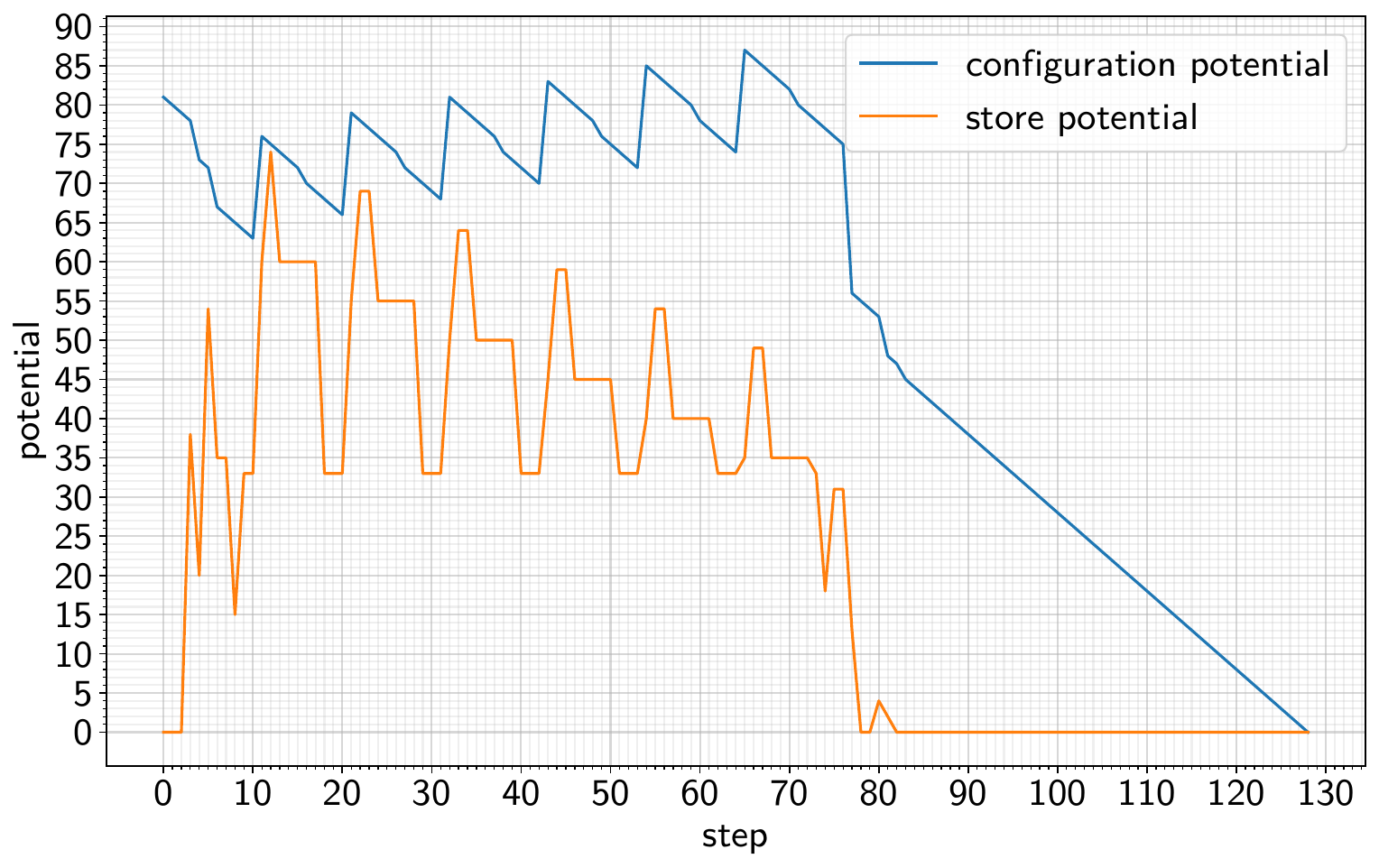}
\caption{Plot of potentials for $\tapp {\tapp {c_6} {\dub}} {(\tlam x {\tapp I x})}$}
\label{fig:plot_of_implosion}
\end{center}
\end{figure}

To see how potential can grow, consider the term
$\tapp {\tapp {c_6} {\dub}} {(\tlam x {\tapp I x})}$,
which belongs to the term family presented in Table~\ref{tab:numbers}.
Its execution trace is plotted in Figure~\ref{fig:plot_of_implosion}.
When compared with the corresponding plot in~\cite{Biernacka-al:PPDP21},
it exhibits similar phases of the computation: initial decrease,
followed by saw-tooth-like increase, and finally, a steady discharge.
In accordance with Table~\ref{tab:numbers}, the RKNL machine performs
$18\cdot 6 + 20 = 128$ transitions on this term.
The first two {$\beta$-transitions} (steps 4 and 6) instantiate
the arguments of the Church numeral $c_6$, but they do not increase the
overall potential.
The potential is increased by {$\beta$-transitions} at steps 11, 21, 32, 43, 54,
and 65, but each increase is bounded by 
$\Phit(\dub) = 20$, since $\dub$ is the function applied in those steps.

\begin{figure}[h]
\begin{center}
\includegraphics[scale=.4]{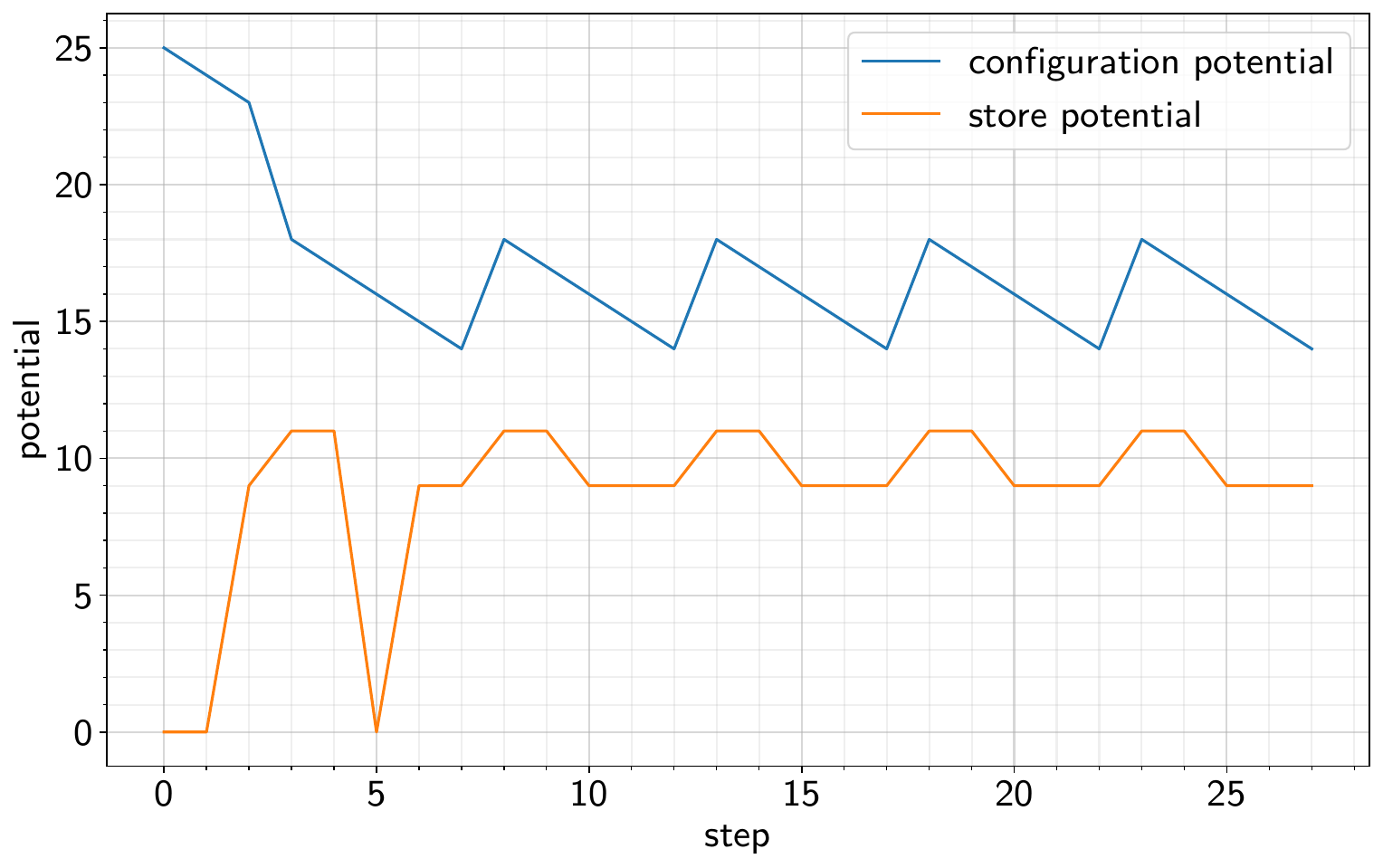}
\caption{Plot of potentials for initial configurations of the evaluation
of the term $\Omega$}
\label{fig:plot_of_omega}
\end{center}
\end{figure}

Finally, arbitrarily many periodic increases of the potential can be observed,
for example, in the normalization of the divergent term $\Omega$.
The start of this periodic behaviour is
presented in Figure~\ref{fig:plot_of_omega}.

\subsection{Completeness}

The final ingredient in the correctness argument is to rule out stuck configurations and 
silent loops. We capture this in the following lemma.

\begin{lemma}[step completeness]\label{lem:step_completeness} If $k$
  is a reachable configuration and $\deck k \stackrel{\no}\to t'$,
  then there exists $k' \neq k$ such that
  $k \twoheadrightarrow k'$
  and $t' \stackrel{\no}{\twoheadrightarrow}=_\alpha \deck {k'}$.
\end{lemma}
\begin{proof}
  The pattern matching of the machine is exhaustive, which is best
  visible in Table~\ref{fig:shape_invariant}, and thus the machine
  can never get stuck.
  Since $\deck k$ is not a normal form (of normal order), by Lemma~\ref{lem:unload},
  $k$~is not a terminal configuration, and thus $k$ has to perform a transition.
  Lemmas~\ref{lem:overhead} and~\ref{lem:alpha} imply that the machine cannot
  reach the terminal configuration without performing one
  of the transitions~(\ref{tr:4}), (\ref{tr:6}) or~(\ref{tr:8}). By
  Lemma~\ref{lem:decrease}, such a transition will be performed
  after a finite number of steps; let $k'$ be the resulting
  configuration. Since normal order is deterministic, Lemma~\ref{lem:step_soundness} guarantees
  $t' \stackrel{\no}{\twoheadrightarrow}=_\alpha \deck {k'}$, as required.
\end{proof}

\noindent
The lemma can be extended to the existence of a normal form, as follows:

\begin{theorem}[completeness]
  If a (possibly open) term $t_0$ reduces in zero or more steps to a
  normal form $t$ (i.e.,
  $t_0 \, {\twoheadrightarrow_\beta}\, t \!\nrightarrow_\beta$), then
  the machine starting from $t_0$ computes $t$ (i.e., there exist
  $t'$ and $\sigma$ such that $t =_\alpha t'$ and
  $\econf \nil {t_0} \nil \nil \twoheadrightarrow \cconf \sigma {t'}
  \nil$).
\end{theorem}
\begin{proof}
  Because normal order is a complete strategy, we know from
  $t_0 \, {\twoheadrightarrow_\beta}\, t \!\nrightarrow_\beta$ 
  that
  $t_0 \, {\stackrel{\no}\twoheadrightarrow}\, t
  \!\nrightarrow_\beta$. Normal order is deterministic, so we can
  iterate Lemma~\ref{lem:step_completeness} along the normal-order
  reduction path.  Once the machine decodes to a normal form, the potential function
decreases strictly, ensuring that the terminal configuration is
reached after finitely many further steps.
\end{proof}

\subsection{Connection with weak call by need}
\label{sec:conservative}
In this section, we show that the strategy realized by our machine is a
conservative extension of the standard weak call-by-need evaluation. We
show that the machine RKNL extends the lazy variant KL of the Krivine
machine, which, according to~\citet{Danvy-Zerny:PPDP13}, is a canonical
implementation of the weak call-by-need evaluation strategy using actual
substitutions. KL is displayed in Table~\ref{fig:lazy_Krivine}.

\begin{table}[htb]
\caption{KL: lazy version of the Krivine abstract machine from~\cite{Danvy-Zerny:PPDP13}}
\label{fig:lazy_Krivine}
\begin{alignat*}{3}
\setcounter{equation}{0}
\mathit{Terms}       \ni && t   &::= x \alt \tapp{t_1}{t_2} \alt \tlam x t\\
\mathit{Values} \ni && v   &::= {\tlam x t} \\
\mathit{Contexts} \ni &\,& E   &::= \hole \alt \tapp E t  \alt {\Kfcache
  x}{E}\\
\mathit{Stores} \ni && \sigma   &::= \varepsilon \alt \sigma[x=t]\\
\mathit{Confs} \ni && k   &::=
  \termconf \sigma t E
  \alt \contconf \sigma v E\\
\mathit{Transitions:}\\[-2em]
\end{alignat*}%
\begin{align}
                            t &\mapsto \termconf \varepsilon t \Knil \nonumber\\
\termconf \sigma {\tapp{t_1}{t_2}} E &\to \termconf  \sigma  {t_1} {\Kcons{\frapp{t_2} } {E}} \tag{1}\label{kl:1}\\
\termconf \sigma {\tlam x t}     E &\to \contconf  \sigma  {\tlam x t}  E  \tag{2}\label{kl:2}\\
\termconf \sigma x               E &\to \termconf  \sigma  {t}   {\Kcons{\fcache
                                x}{E}} \;\;: t=\sigma (x), \; t\not \in \mathit{Values}
                                 \tag{3}\label{kl:3}\\
\termconf \sigma x               E &\to \contconf  \sigma {v} E
                                {\hspace{11mm}}: v = \sigma (x) 
                                   \tag{4}\label{kl:4}\\
  \contconf \sigma v {\Kcons {\fcache x} E}
                              &\to \contconf {\update \sigma x v} v E  \tag{5}\label{kl:5}\\
  \contconf \sigma {\tlam x t} {\Kcons{\frapp{t_2}}{E}}
                              &\to \termconf {\Kalloc \sigma {x'} {t_2}} {\Kupdate t x {x'}} E \;\;\;\;: x'\not\in\mathrm{dom}(\sigma)  \tag{6}\label{kl:6}\\
\contconf \sigma v \Knil &\mapsto \ansconf  \sigma v \nonumber
\end{align}
\end{table}

The relation between RKNL and KL is described by a direct
correspondence between the first six transitions of RKNL and
transitions of KL. Intuitively, for $i=1,\ldots,6$, as long as a value
is not reached, transition ($i$) of RKNL does the same as
transition ($i$) of KL, only using different term representation. The
essential difference is that RKNL is environment-based while KL
is substitution-based. 

Substitution in KL occurs only in transition~(\ref{kl:6}), where it is
used to $\alpha$-rename a variable with a fresh one representing a
location on the store. In RKNL, on the other hand, a separate
syntactic category of store locations is introduced, and the mapping
from variables to store locations is kept in the environment. Other,
inessential differences are that RKNL uses stack to represent contexts
(which are directly used in KL) and it stores some additional
information not used by KL. In the rest of this section we formalize
this correspondence.

Recall that weak call by need is a reduction strategy meant to
evaluate only closed terms.
For a given closed term $t$ by the \emph{weak evaluation} of $t$ we
mean the longest sequence of successive configurations of RKNL that
starts with loading $t$ and uses only transitions (\ref{tr:1}) to
(\ref{tr:6}). We refer to configurations occurring in this sequence as
to \emph{weak configurations}.

Table~\ref{fig:translation} shows the translation $\Kdeck{\cdot}$ of weak
configurations of RKNL to configurations of KL. It is based on
translations $\Kdecc{\cdot}$ of closures, $\Kdecs{\cdot}$ of stacks,
$\Kdecv{\cdot}$ of values and $\Kdecst{\cdot}$ of stores.

\begin{table}[ht]%
\caption{Translation of weak configurations of RKNL to configurations of
  KL}
\label{fig:translation}
\begin{equation*}
\begin{aligned}
\Kdecc{\clos x e} &:= \lookup e x\\
\Kdecc{\clos {\tapp {t_1} {t_2}} e} &:= \tapp {\Kdecc{\clos {t_1} e}} {\Kdecc{\clos {t_2} e}}\\
\Kdecc{\clos {\tlam x t} e} &:= \tlam x {\Kdecc{\clos t {\update e x x}}}\\
\end{aligned}
\qquad
\begin{aligned}[c]
\Kdecs{\nil} &:=\; \hole\\
\Kdecs{\cons {\frapp c} s}     &:=\; \plug {\Kdecs{s}} {\frapp{\Kdecc{ c}}}\\
\Kdecs{\cons {\fcache \ell} s} &:=\; \plug {\Kdecs{s}} {{\Kfcache \ell} \hole}
\end{aligned}
\end{equation*}

\begin{equation*}
  \begin{aligned}[c]
\Kdecv{\abs {\tlam x t} e \ell} := \Kdecc{\clos {\tlam x t} e}\\
\Kdecst{\sigma}(\ell):=\begin{cases}\Kdecc{c} \;:\sigma(\ell)= \tobedone c\\
                                    \Kdecv{v} \;:\sigma(\ell)= \done v\end{cases}
\end{aligned}
\qquad
\begin{aligned}[c]
 \Kdeck{\econf \sigma t e s} &:= \termconf {\Kdecst{\sigma}} {\Kdecc{\clos t e}}
 {\Kdecs s} \\
\Kdeck{\cconf \sigma   v s} &:= \contconf {\Kdecst{\sigma}} {\Kdecv{v}} {\Kdecs{s}}
\end{aligned}
\end{equation*}
\end{table}

The correctness of the translation follows from  several observations:
\begin{itemize}
\item All values occurring in a weak evaluation are annotated lambda
  abstractions of the form $\abs {\tlam x t} e
  \ell$. This is because free variables introduced by
  transition~(\ref{tr:4}) do not occur in closed terms,
  and all other types of values are introduced in
  transitions~(\ref{tr:7}),~(\ref{tr:10}) and~(\ref{tr:11}), which are
  not used in weak evaluations. Therefore, all values occurring in a
  weak evaluation can be translated with~$\Kdecv{\cdot}$.
\item Frames $\flapp t$ and $\flam
  x$ are introduced by transitions~(\ref{tr:9}) and~(\ref{tr:7}) and
  thus do not occur in weak evaluations. Therefore, all stacks
  occurring in a weak evaluation can be translated with $\Kdecs{\cdot}$.
\item The last configuration in a weak evaluation is a
  $\ctrian$-configuration. This is because in every
  $\etrian$-configuration one of
  transitions~(\ref{tr:1})--(\ref{tr:4}) is fireable.
\item The stack in the last configuration in a weak evaluation is
  empty. This is because with a nonempty stack a
  transition~(\ref{tr:5}) or~(\ref{tr:6}) is fireable. Therefore the
  translation of the last configuration in a weak evaluation is
  directly unloaded to an answer of KL.
\end{itemize}

The following lemma can be proved  by an easy induction.
\begin{lemma}
  If $k\to k'$ is a transition in a weak evaluation of a closed term in RKNL,
  then $\Kdeck{k}\to\Kdeck{k'}$ is a transition of KL.
\end{lemma}
In consequence, the evaluation of a closed term by KL simulates a weak
evaluation of the same term by RKNL.  Having in mind that both
machines are deterministic and that both sequences stop at the same
moment, this simulation is a bisimulation. Since weak evaluation is
only an initial part of an evaluation by RKNL, we obtain that RKNL
conservatively extends KL. In other words, the strategy realized by
RKNL is a conservative extension of the call-by-need strategy.

A fuller picture could be drawn by relating RKNL to the Strong CbNd
strategy of~\citet{Balabonski-al:ICFP17}, which is
itself a conservative extension of weak call-by-need. A direct
comparison is non-trivial because the Balabonski et al. strategy
operates on a language with explicit substitutions, and any formal
correspondence would require a translation between the two syntactic
frameworks. A formal characterization of the strategy realized by RKNL
independently of its machine description remains an open problem we
identify as a natural direction for future work.

\section{Conclusion}
\label{sec:conclusion}
We have presented a simple and efficient abstract machine for the
Strong Call-by-Need strategy in the lambda calculus. The machine has
been derived from a no\-rma\-li\-zation-by-evaluation higher-order
functional program by means of a series of off-the-shelf program
transformation techniques, and it has been carried out in Racket,
almost automatically. We proved the expected properties of the
resulting machine: its soundness and completeness with respect to
normal order, and its reasonability. Specifically, we proved that the
number of steps in an execution is bilinear in the number of
$\beta$-steps in the corresponding reduction sequence and in the size
of the initial term.  In the proofs we use a ghost abstract machine
and potential-function techniques.

This work also confirms the versatility of the derivational approach
to semantics; its additional benefit is the relative simplicity of the
design and of reasoning about abstract machines. As future work, we
plan to adapt the present methodology to deconstruct and find
connections between other existing artefacts of strong call by need,
in particular Cr{\'e}gut's lazy KNL machine \citep{Cregut:HOSC07},
a~new, non-conservative strong call-by-need strategy proposed recently
by \citet{Balabonski-al:FSCD21}, and Accattoli and Leberle's study of
useful sharing by explicit substitutions in call by need
\citep{DBLP:conf/csl/AccattoliL22}.

\section*{Data availability statement}

The code accompanying this paper is available at the Zenodo repository \citep{AbstractMachinesWorkshop}.

\paragraph*{Acknowledgements.}
This paper is a revised and extended version of \cite{BiernackaCD22}.
The research is supported by the National Science
Centre of Poland, under grant number 2019/33/B/ST6/00289.
We thank Dariusz Biernacki and the anonymous reviewers
of the paper and of the software artefact for their helpful comments,
Maciej Buszka for making {\tt semt} available to us for
experimentation, Beniamino Accattoli for the encouragement to work on
strong call-by-need strategy, Tomasz Wierzbicki for teaching a course
based on Chris Okasaki's book \cite{Okasaki:99} at the University of
Wrocław, Kim-Ee Yeoh and other participants of ICFP for their
questions.

\bibliographystyle{ACM-Reference-Format}
\bibliography{mybib} 
\label{lastpage01}

\end{document}

%% file: macros.tex
\newcommand{\ie}{i.e.}
\newcommand{\eg}{e.g.}
\newcommand \alt{\;\;|\;\;}
\newcommand \subtype{<:}

\newcommand \cbn{\mathit{cbn}}
\newcommand \no{\mathit{no}}
\newcommand \NO{\mathit{NO}}
\newcommand \pred{\mathit{pred}}
\newcommand \pair{\mathit{pair}}
\newcommand \dub{\mathit{dub}}
\newcommand{\KI}{\rotatebox[origin=c]{180}{$K$}}
\newcommand{\rknle}{\text{RKNL$\!^{\smallsetminus(\ref{tr:8})}$}}

\newcommand{\nil}{[\,]}
\newcommand{\cons}[2]{{#1 :: #2}}

\newcommand{\tobedone}[1]{{#1}_{\mbox{\tiny \textcolor{red!80!black}\XSolidBrush}}}
\newcommand{\done}[1]{{#1}_{\mbox{\tiny\textcolor{green!50!black}\Checkmark}}}

\newcommand{\fun}[2]{{{#1} \to {#2}}}
\newcommand{\call}[2]{{{#1}({#2})}}
\newcommand{\lookup}[2]{{\call {#1} {#2}}}
\newcommand{\alloc}[3]{{{#1} \!\ast\! [{#2} \!\mapsto\! {#3}]}}
\newcommand{\update}[3]{{{#1}[{#2} \!:=\! {#3}]}}
\newcommand{\fresh}[1]{{\check{#1}}}
\newcommand{\plug}[2]{ {{#1}[{#2}]} }
\newcommand{\cut}[1]{ \text{\hl{${\big\langle{#1}\big\rangle}$}} }
\newcommand{\cache}[2]{{#1}\!:=\!{#2}}

\newcommand{\tlam}[2]{{\lambda{#1}. {#2}}}
\newcommand{\tapp}[2]{{{#1}\,{#2}}}
\newcommand{\clos}[2]{{( #1, #2 )}}
\newcommand{\abs}[3]{\cache{#3} {\clos{#1}{#2}}}

\newcommand{\hole}{\Box}
\newcommand{\flapp}[1]{{\tapp{#1}{\hole}}}
\newcommand{\frapp}[1]{{\tapp{\hole}{#1}}}
\newcommand{\frclos}[2]{\frapp{\clos{#1}{#2}}}
\newcommand{\fcache}[1]{\cache{#1}{\hole}}
\newcommand{\flam}[1]{{\tlam {#1} \hole}}

\newcommand \etrian{{\textcolor{blue}\triangledown}}
\newcommand \ctrian{{\textcolor{green!50!black}\vartriangle}}
\newcommand{\econf}[4]{{\langle {\clos {#2}  {#3}}, {#4}, {#1}\rangle_\etrian}}
\newcommand{\cconf}[3]{{\langle{#2}, {#3}, {#1}\rangle_\ctrian}}
\newcommand{\cpconf}[3]{{\langle{#2}, {#3}, {#1}\rangle_{\ctrian'}}}

\newcommand{\Kcons}[2]{{#2[#1]}}
\newcommand{\Kfcache}[1]{{#1}\!:=}
\newcommand{\Knil}{\Box}
\newcommand{\Kdec}[1]{\overline{#1}}
\newcommand{\Kdecc}[1]{\Kdec{#1}^{\mathsf{c}}}
\newcommand{\Kdeck}[1]{\Kdec{#1}^{\mathsf{k}}}
\newcommand{\Kdecv}[1]{\Kdec{#1}^{\mathsf{v}}}
\newcommand{\Kdecs}[1]{\Kdec{#1}^{\mathsf{s}}}
\newcommand{\Kdecst}[1]{\Kdec{#1}^{\mathsf{st}}}
\newcommand{\Kupdate}[3]{{{#1}[{#3}/{#2}]}}
\newcommand{\Kalloc}[3]{{{#1} [{#2} = {#3}]}}

\newcommand{\termconf}[3]{{\langle{#2}, {#3}, {#1}\rangle_{\mathit{term}}}}
\newcommand{\contconf}[3]{{\langle{#3}, {#2}, {#1}\rangle_{\mathit{cont}}}}
\newcommand{\ansconf}[2]{{\langle{#2}, {#1}\rangle_{\mathit{ans}}}}

\newcommand{\Phit}{\Phi_\text{t}}
\newcommand{\Phiv}{\Phi_\text{v}}
\newcommand{\Phis}{\Phi_\text{s}}
\newcommand{\Phih}{\Phi_\sigma}
\newcommand{\Phik}{\Phi_\text{k}}

\newcommand{\decoding}[2]{\underline{#1}_{\,#2}}
\newcommand{\decc}[1]{{\decoding {#1} {\mathsf{c}}}}
\newcommand{\decv}[1]{{\decoding {#1} {\mathsf{v}}}}
\newcommand{\decs}[1]{{\decoding {#1} {\mathsf{s}}}}
\newcommand{\deck}[1]{{\decoding {#1} {\mathsf{k}}}}

\newcommand{\nf}{n}
\newcommand{\neu}{a}

\newcommand{\rred}[2]{ {\stackrel{#1}{\twoheadrightarrow}_{#2}} }

\newcommand{\subst}[3]{\update{#3}{#1}{#2}}